\documentclass[journal]{IEEEtran}
\usepackage{graphicx}
\usepackage{amssymb}
\usepackage{amsthm}
\usepackage[noadjust]{cite}%cite
\usepackage{setspace}               % set line space
\usepackage{stfloats}
\usepackage{bbding}
\usepackage{amsthm}
\usepackage{amsmath}
\usepackage{flushend}
\usepackage{cases,subeqnarray}
\usepackage{bm,multirow,bigstrut}
\usepackage{textcomp}
\usepackage{latexsym,bm}
\usepackage{booktabs}
\usepackage{xcolor}
\usepackage{mathtools}
\usepackage{dsfont}
\usepackage{extarrows}
\usepackage{mathrsfs}
\usepackage{subfigure}
\usepackage{cite}
\usepackage{soul}
\usepackage{bm}
\usepackage{cleveref}
\usepackage{multicol}       % table
\usepackage{multirow}       % table
\usepackage{array}          % table
\usepackage{colortbl}
\usepackage{makecell}
\definecolor{crimson}{RGB}{192,0,0}         % color crimson
\definecolor{navy}{RGB}{47,85,151}         % color crimson

\makeatletter                               % algorithm
\newif\if@restonecol
\makeatother

\makeatletter
\newif\if@restonecol
\makeatother

\usepackage[linesnumbered,ruled,vlined]{algorithm2e}%[ruled,vlined]{
\usepackage{algpseudocode}
\usepackage{amsmath}
\usepackage{stfloats}
\usepackage{color}
\usepackage{placeins}
\usepackage{tabularx,colortbl}
\usepackage{amssymb}

\usepackage{listings}                       % Matlab code
\usepackage[framed,numbered,autolinebreaks,useliterate]{mcode}
\theoremstyle{plain}

\newtheorem{lemm}{Lemma}
\newtheorem{prop}{Proposition}

\theoremstyle{plain}
\newtheorem{rem}{Remark}

\IEEEoverridecommandlockouts

\usepackage[ruled,vlined]{algorithm2e}
\SetAlgoSkip{}            % 取消算法前后额外竖直间距
\SetAlgoInsideSkip{}      % 如仍觉得算法内部上下留白偏大，可一起关掉（可选）

\begin{document}

%----------------------------title&author&thanks----------------------------

\title{Sparse Activation for Sustainable Cell-Free Massive MIMO Networks: Less is More
\thanks{Z. Wang and E. Bj{\"o}rnson are with KTH Royal Institute of Technology, Stockholm, Sweden, and also with Digital Futures, Stockholm, Sweden (e-mail: zhewang2@kth.se, emilbjo@kth.se); S. Chen (corresponding author) is with Purple Mountain Laboratories, Nanjing, China, and with Southeast University, Nanjing, China (e-mail: shuaifeichen@seu.edu.cn).}}
% \thanks{S. Chen (corresponding author) is with Purple Mountain Laboratories, Nanjing, China, and with Southeast University, Nanjing, China (e-mail: shuaifeichen@seu.edu.cn).}}
\author{Zhe Wang,~\IEEEmembership{Member,~IEEE,} Shuaifei Chen,~\IEEEmembership{Member,~IEEE,} and Emil Bj{\"o}rnson,~\IEEEmembership{Fellow,~IEEE}\vspace*{-0.7cm}}
\maketitle

%----------------------------abstract----------------------------
\begin{abstract}

Motivated by the vision of making sixth-generation (6G) networks sustainable, we study the sparse antenna/array activation problems in uplink cell-free massive multiple-input multiple-output (CF mMIMO) networks. We first develop an antenna-level optimal bilinear equalizer (OBE) weighting framework, in which each access point-user equipment (AP-UE) pair is assigned a matrix-valued long-term weight to shape the contribution of individual antenna elements, thereby generalizing the conventional large-scale fading decoding (LSFD) strategy from scalar coefficients to antenna-element-aware weighting. Building on this structure, we formulate sparse antenna activation as structured sparsity-inducing mean square error (MSE) minimization problems, and design four activation schemes at two granularities: antenna-level and array-level, each with UE-specific and network-wide (all-UEs) variants. The resulting convex problems are solved efficiently via the proximal method with closed-form group-wise updates, while the network-wide schemes are modeled through hierarchical sparsity and handled by a tree-structured proximal operator. Numerical results under correlated Rician channels and a detailed power consumption model demonstrate that the OBE weighting scheme consistently improves spectral efficiency over the LSFD, with gains increasing with the number of antennas. Meanwhile, the studied sparse activation schemes can achieve substantial energy efficiency improvement and power reduction with controllable spectral efficiency loss.

\end{abstract}
%----------------------------keywords----------------------------

%\newpage
\IEEEpeerreviewmaketitle

\section{Introduction}
\bstctlcite{BSTcontrol}
The transition from the fifth-generation (5G) to the sixth-generation (6G) wireless networks is envisioned to deliver not only unprecedented communication capabilities but also a paradigm shift toward green and sustainable communications, where resource-efficient network operation becomes a primary design objective \cite{ITU,6G-IA,10003083}. In particular, the IMT-2030 vision emphasizes energy-efficient and low-power technologies and advocates responsible resource usage under circular-economy principles, which jointly elevate sustainability from an auxiliary constraint to a fundamental requirement in 6G system design \cite{ITU}. Multiple-input multiple-output (MIMO) technology is a cornerstone technology aligned with this grand vision, as it enhances spectral efficiency and link reliability by exploiting spatial multiplexing and beamforming gains \cite{5595728,9113273,bjornson2019massive,11155198}. Notably, the past decade has witnessed a clear ``antenna abundance" evolution trend for MIMO, where the performance of wireless networks is continuously pushed forward by scaling up antenna arrays to massive numbers, to leverage larger spatial multiplexing gains. This evolution, while highly effective in improving capacity, also implies a growing hardware footprint and increasingly complex signal processing, which directly interacts with the sustainability objectives of future 6G systems.

On the other hand, future MIMO evolution can also be posed around antenna flexibility and adaptability, where antenna configurations and network topologies are dynamically reconfigured to match the propagation and traffic conditions, rather than continuously increasing the antenna count \cite{wang2025flexible}.
Among many representative next-generation MIMO technologies in \cite{wang2025flexible}, cell-free massive MIMO (CF mMIMO) can be viewed as a flexible-deployment topology, where many geographically distributed access points (APs) coherently serve users under central assistance from the central processing unit (CPU) \cite{7827017,cell-free-book,OBETrans,chen2025channel}. This architecture exploits an abundance of APs to achieve macro-diversity, but it can be harnessed by activating only a tailored subset of AP arrays and antennas depending on the current user and service demand, thereby facilitating energy-efficient and low-power consumption operations. A straightforward implementation of CF mMIMO assumes ``always-on, full cooperation", where all APs and antenna elements continuously participate in serving all user equipments (UEs). This is not necessarily green and sustainable, since the energy budget spent on circuit activity, baseband processing, and fronthaul/CPU overhead scales with the number of active antennas/APs, even when many links contribute negligible marginal rate improvements due to pathloss and spatial heterogeneity. 

For a given UE and service demand, the energy-efficient and scalable serving set is typically a small subset of APs, since involving far-away APs often yields marginal performance gains but incurs non-negligible processing and fronthaul/transport overhead \cite{9064545,10380322}. This principle is formalized in scalable CF mMIMO in \cite{9064545} via dynamic cooperation clusters, where each UE is served by a tailored subset of APs and many AP-UE links are switched off with minor performance loss. This naturally motivates moving from the conventional full activation principle in CF mMIMO to selective activation, where only a tailored subset of AP arrays or even antenna elements within each AP is kept active to serve each UE. Importantly, the key research question is to \emph{determine which AP arrays/antennas should be activated for which UE}, while meeting the intended per-UE spectral efficiency (SE) or rate requirements; otherwise, the most energy-minimizing solution would be to switch off the entire infrastructure. For instance, the authors in \cite{11264456} proposed and designed a promising joint pilot and data transmission framework for grant-free random access, which can effectively reduce the overhead for users joining in CF mMIMO networks and guarantee the system performance of the active users. The authors in \cite{11391509} delivered a scalable and sustainable vision in CF mMIMO networks, where each UE is connected with only a small optimized set of serving APs while achieving excellent performance based on the efficient mitigation of unnecessary interference. Moreover, the authors in \cite{10505156} studied fairness-oriented scheduling in scalable CF mMIMO networks, where the active-user set is dynamically selected to balance spectral efficiency and long-term user fairness under extremely heavy user load.

Building on this selective activation idea, some papers have started to design CF mMIMO networks with sparse activation/association to improve energy efficiency (EE). In particular, the authors in \cite{10113887} investigate the sparse large-scale fading decoding (LSFD) and large-scale fading precoding (LSFP) schemes, where the AP-UE serving links are selected through the sparsity-inducing mean square error (MSE) minimization so that each UE is supported by only a subset of APs with minor SE degradation and significant EE enhancement. These approaches are built on the LSFD-type two-layer processing, where each AP first forms a local soft estimate, and then, the CPU weights these estimates using per-AP-UE scalar long-term coefficients. Nevertheless, the scalar nature of the second-layer coefficient provides only AP-level processing, without the ability to customize the weighting across the multiple antenna elements within each AP. Consequently, these LSFD-based sparse designs mainly operate at the AP-level and cannot directly tailor the activation of individual antenna elements within an AP on a per-UE basis, thereby limiting fine-grained antenna/array activation and underutilizing the spatial processing potential offered by multi-antenna APs.
Thus, it is vital to develop an antenna-level weighting scheme that enables fine-grained, UE-specific activation of antenna elements and arrays, thereby improving energy efficiency and reducing overall power consumption while maintaining the intended SE performance.

Motivated by the above observations, in this paper, we first construct a novel antenna-level weighting scheme for CF mMIMO networks, and then, we present four antenna/array activation schemes with the aid of sparse optimization. By doing so, these schemes can facilitate the future green and sustainable wireless communication vision from the perspective of efficient antenna/array resource scheduling. The major contributions of this paper are listed as follows.   
\begin{itemize}
\item We develop an antenna-level optimal bilinear equalizer (OBE) weighting framework, where a matrix-valued long-term weight is introduced to each AP-UE pair to shape the contributions of the individual antenna elements. This generalizes the LSFD weighting scheme from scalar coefficients to matrix weights, thereby enabling antenna-element-specific weighting and fine-grained multi-antenna control.
\item Building upon the antenna-level weighting structure, we formulate antenna-level activation as sparsity-inducing MSE minimization problems that selectively deactivate particular antenna elements through structured regularization. We study both UE-specific activation and network-wide activation across all UEs, where the latter is characterized by a hierarchical sparsity model and solved efficiently via the tree-structured proximal method.
\item We further study the array-level sparse activation framework based on the sparsity-inducing MSE minimization. We consider both UE-wise array activation and network-wide array activation across all UEs, where the entire array can be selected and switched on/off accordingly. Numerical results validate the effectiveness of the studied sparse activation schemes on the EE enhancement and power consumption reduction and deliver many useful insights for the practical implementation of these schemes.
\end{itemize}

\textbf{\emph{Notation}}: Boldface lowercase letters $\mathbf{x}$ and boldface uppercase letters $\mathbf{X}$ denote vectors and matrices, respectively. The operators $\mathbb{E}\{\cdot\}$, $(\cdot)^T$, $(\cdot)^H$, $\mathrm{vec}(\cdot)$, $\mathrm{tr}(\cdot)$, $\Re(\cdot)$, and $\Im(\cdot)$ denote the expectation, transpose, conjugate transpose, vectorization, trace, real part, and imaginary part, respectively. $\mathbf{I}_N$ denotes the $N\times N$ identity matrix, $\otimes$ denotes the Kronecker product, and $\mathrm{diag}(\mathbf{X}_1,\ldots,\mathbf{X}_n)$ denotes a block-diagonal matrix with diagonal blocks $\mathbf{X}_1,\ldots,\mathbf{X}_n$. $\|\cdot\|_2$ denotes the Euclidean norm and $|\cdot|$ denotes the absolute value of a scalar.

\vspace{-0.5cm}

\section{System Model}

In this paper, we study a cell-free massive MIMO network, where $M$ APs equipped with uniform planar arrays (UPAs) serve $K$ single-antenna UEs in a coverage area. All APs are connected to the CPU via fronthaul links, and each AP UPA has $N=N_HN_V$ antennas, where $N_H$ and $N_V$ denote the number of antennas at the horizontal and vertical directions, respectively. Let $\mathbf{h}_{mk}\in \mathbb{C} ^N$ denote the channel response between the $m$-th array and the $k$-th UE, which is assumed to remain constant in time-frequency blocks that consist of $\tau_c$ symbols, based on the standard block fading model \cite{cell-free-book}. We assume that $\mathbf{h}_{mk}$ is independently and identically distributed across blocks, and $\mathbf{h}_{mk}$ for different array-UE pairs are also independent. We study the correlated Rician fading channel model, where $\mathbf{h}_{mk}$ can be modeled as $\mathbf{h}_{mk}=\overline{\mathbf{h}}_{mk}e^{j\theta _{mk}}+\check{\mathbf{h}}_{mk}$ with $\overline{\mathbf{h}}_{mk}\in \mathbb{C} ^N$ and $\check{\mathbf{h}}_{mk}\in \mathbb{C} ^N$ being the deterministic line-of-sight (LoS) and stochastic non-line-of-sight (NLoS) components, respectively. Moreover, $\theta _{mk}\sim \mathcal{U} ( -\pi ,\pi ]$ denotes the random phase-shift for the LoS component between AP $m$ and UE $k$. Then, we model the NLoS component as $\check{\mathbf{h}}_{mk}\sim \mathcal{N} _{\mathbb{C}}\left( \bm{0},\mathbf{R}_{mk} \right)$, where $\mathbf{R}_{mk} \in \mathbb{C} ^{N\times N} $ denotes the spatial correlation matrix between AP $m$ and UE $k$.

\begin{figure*}[t]
{{\begin{align}\tag{7}\label{SINR_k}
\mathrm{SINR}_k=\frac{p_k \left|\sum\limits_{m=1}^M\mathbb{E} \{\mathbf{v}_{mk}^{H}\mathbf{h}_{mk}\} \right|^2}{\sum\limits_{l=1}^K p_l\mathbb{E} \left\{ \left|\sum\limits_{m=1}^M\mathbf{v}_{mk}^{H}\mathbf{h}_{ml} \right|^2 \right\}-p_k \left|\sum\limits_{m=1}^M\mathbb{E} \{\mathbf{v}_{mk}^{H}\mathbf{h}_{mk}\} \right|^2+\sigma ^2\sum\limits_{m=1}^M\mathbb{E} \{\| \mathbf{v}_{mk}\|^2\}}
\end{align}}
\hrulefill
\vspace*{-0.1cm}
%\vspace*{3pt}
}\end{figure*}

During the channel estimation phase, we utilize $\tau _p$ orthogonal pilot signals. We define $\mathcal{P} _k$ as the subset of UEs that apply the same pilot signal as UE $k$ (including UE $k$ itself). Following the standard MMSE channel estimation methodology as in \cite{cell-free-book}, we can obtain the phase-aware MMSE estimate of $\mathbf{h}_{mk}\in \mathbb{C} ^N$ at AP $m$ as\footnote{``Phase-aware" means that the APs are assumed to know the instantaneous LoS phase-shifts within each coherence block. These phase-shifts remain constant over the frequency domain but vary randomly over the time domain at the pace of the small-scale fading. Hence, it is reasonable to treat them as blockwise known parameters for channel estimation.}
\vspace{-0.55em}\begin{equation}
\widehat{\mathbf{h}}_{mk}=\overline{\mathbf{h}}_{mk}e^{j\theta _{mk}}+\sqrt{p_{k,p}}\mathbf{R}_{mk}\mathbf{\Psi }_{mk}^{-1}\left( \mathbf{y}_{mk}^{p}-\overline{\mathbf{y}}_{mk}^{p} \right),
\vspace{-0.3em}\end{equation}
where $\mathbf{y}_{mk}^{p}=\sum_{l\in \mathcal{P} _k}{\sqrt{p_{l,p}}\tau _p\mathbf{h}_{ml}}+\mathbf{n}_{mk}^{p}$, $\overline{\mathbf{y}}_{mk}^{p}=\sum_{l\in \mathcal{P} _k}{\sqrt{p_{l,p}}\tau _p\overline{\mathbf{h}}_{ml}}e^{j\theta _{ml}}$, $\mathbf{n}_{mk}^{p}\sim \mathcal{N} _{\mathbb{C}}( \mathbf{0},\tau _p\sigma ^2\mathbf{I}_N ) $, and $
\mathbf{\Psi }_{mk}=\sum_{l\in \mathcal{P} _k}{p_{l,p}\tau _p\mathbf{R}}_{ml}+\sigma ^2\mathbf{I}_N$.  Moreover, $p_{k,p}$ denotes the transmit power of UE $k$ during the pilot transmission phase and $\mathbf{n}_{mk}^{p}$ denotes projected additional noise with $\sigma ^2$ being the noise power. Following this estimator, the channel estimation error $\tilde{\mathbf{h}}_{mk}=\mathbf{h}_{mk}-\widehat{\mathbf{h}}_{mk} $ is conditionally independent of the channel estimate with $\mathrm{Cov} \{ \tilde{\mathbf{h}}_{mk} |\theta _{mk}\} =\mathbf{C}_{mk}=\mathbf{R}_{mk}-p_{k,p}\tau _p\mathbf{R}_{mk}\mathbf{\Psi }_{mk}^{-1}\mathbf{R}_{mk}$. More derivation details for this phase-aware MMSE channel estimator can be found in \cite{OBETrans}.

\vspace{-0.4cm}

\section{OBE Weighting-Based Data Transmission}
In this section, we study the uplink data transmission, where a novel OBE weighting strategy is proposed to further enhance the achievable SE performance compared to the widely studied LSFD weighting strategy in \cite{cell-free-book}. During the data transmission phase, all UEs transmit $\tau_u=\tau_c-\tau_p$ data symbols, and the received data signal at AP $m$ can be represented by 
$\mathbf{y}_m=\sum_{k=1}^K{\mathbf{h}_{mk}x_k}+\mathbf{n}_m\in \mathbb{C} ^N,$
where $x_k\sim \mathcal{N} _{\mathbb{C}}(0,p_k)$ denotes the data symbol sent by UE $k$ with $p_k$ being the transmit power of UE $k$ during the data transmission phase, $\mathbf{n}_m\sim \mathcal{N} _{\mathbb{C}}( \mathbf{0},\sigma ^2\mathbf{I}_N ) $ represents the additive noise at AP $m$. 
Relying on the local decoding vector $\mathbf{v}_{mk}\in \mathbb{C} ^N$ for UE $k$, the $m$-th AP can locally derive the soft estimate of $x_k$ as \vspace{-0.55em}\begin{equation}\widehat{x}_{mk}=\mathbf{v}_{mk}^{H}\mathbf{y}_m\!=\!\mathbf{v}_{mk}^{H}\mathbf{h}_{mk}x_k\!\!+\!\!\sum_{l\ne k}^K{\mathbf{v}_{mk}^{H}\mathbf{h}_{ml}x_l\!}+\!\mathbf{v}_{mk}^{H}\mathbf{n}_m.\vspace{-0.3em}\end{equation} 

\vspace{-0.6cm}

\subsection{Conventional LSFD Weigthing Scheme}
In the conventional LSFD weighting scheme \cite{cell-free-book}, all arrays send their local soft estimates $\widehat{x}_{mk}$ to the CPU, where the CPU implement the LSFD weighting strategy to combine these soft estimates as $\bar{x}=\sum_{m=1}^M{a_{mk}^{*}\widehat{x}_{mk}}$ with $a_{mk}$ being the complex weighting scalar of AP $m$ tailored for UE $k$. The CPU can compute the optimal statistics-based LSFD vector for each UE based on the Rayleigh quotient maximization result as
\vspace{-0.55em}\begin{equation}\label{lsfd_eq}
\mathbf{a}_{k}^{*}\!\!=\!\!\left( \sum_{l=1}^K{p_l\mathbf{\Xi }_{kl}}-p_k\mathbb{E} \left\{ \mathbf{p}_{k} \right\} \mathbb{E} \left\{ \mathbf{p}_{k}^{H} \right\} +\sigma ^2\mathbf{Q}_k \right) ^{-1}\mathbb{E} \left\{ \mathbf{p}_{k} \right\},
\vspace{-0.3em}\end{equation}
where $\mathbf{a}_k=[ a_{1k},\dots ,a_{Mk} ] \in \mathbb{C} ^M$,  $\mathbf{Q}_k=\mathrm{diag}[ \mathbb{E} \{ \| \mathbf{v}_{1k} \| ^2 \} ,\dots ,\mathbb{E} \{ \| \mathbf{v}_{Mk} \| ^2 \} ] \in \mathbb{C} ^{M\times M}$, $\mathbf{p}_{k}=[ \mathbf{v}_{1k}^{H}\mathbf{h}_{1k},\dots ,\mathbf{v}_{Mk}^{H}\mathbf{h}_{Mk} ] ^T\in \mathbb{C} ^M$, and $\mathbf{\Xi }_{kl}\in \mathbb{C} ^{M\times M}$ with its $(m_1,m_2)$-th element being $[ \mathbf{\Xi }_{kl} ] _{m_1m_2}=\mathbb{E} \{( \mathbf{v}_{m_1k}^{H}\mathbf{h}_{m_1l} ) ^H( \mathbf{v}_{m_2k}^{H}\mathbf{h}_{m_2l} ) \}$. The optimal LSFD scheme as \eqref{lsfd_eq} can not only maximize the use-and-then-forget (UatF) capacity bound-measured achievable SE performance as shown in \cite[Eq. (9)]{10113887}, but also minimize $\mathrm{MSE}_{k,\mathrm{LSFD}}\!=\!\mathbb{E} \{|x_k-\bar{x}_k|^2\}$ as shown in \cite[Eq. (14)]{10113887}.

\vspace{-0.3cm}

\subsection{Proposed OBE Weigthing Scheme}
In this paper, we propose a more general weighting strategy, where a matrix-valued weighting scheme can be applied at each array for each UE, going beyond the scalar weighting strategy in the LSFD scheme. More specifically, we assume a weighted local decoding scheme $\mathbf{v}_{mk}$ constructed as 
\vspace{-0.55em}\begin{equation}\label{local_decoding_scheme}
\mathbf{v}_{mk}=\mathbf{W}_{mk}\overline{\mathbf{v}}_{mk},
\vspace{-0.3em}\end{equation}
where $\mathbf{W}_{mk}\in \mathbb{C} ^{N\times N}$ denotes the statistics-based weighting matrix at AP $m$ for UE $k$ and $\overline{\mathbf{v}}_{mk}\in \mathbb{C} ^N$ denotes the local instantaneous channel state information (CSI)-based combining vector at AP $m$ for UE $k$. Two widely applied local combining schemes are the L-MMSE combining 
\vspace{-0.55em}\begin{equation}\label{local_MMSE}
\begin{aligned}
\overline{\mathbf{v}}_{mk}=p_k\left( \sum_{l=1}^K{p_l\left( \widehat{\mathbf{h}}_{ml}\widehat{\mathbf{h}}_{ml}^{H}+\mathbf{C}_{ml} \right)}+\sigma ^2\mathbf{I}_N \right) ^{-1}\widehat{\mathbf{h}}_{mk},
\end{aligned}
\vspace{-0.3em}\end{equation}
and the local-maximum ratio (L-MR) combining $\overline{\mathbf{v}}_{mk}=\widehat{\mathbf{h}}_{mk}$. Then, all arrays transmit the decoding results to the CPU, and the final decoding result can be denoted as
\vspace{-0.55em}\begin{equation}\label{decoding_signal}
\check{x}_k\!\!=\!\!\!\sum_{m=1}^M{\!\hat{x}_{mk}}\!=\!\!\left( \sum_{m=1}^M{\mathbf{v}_{mk}^{H}\mathbf{h}_{mk}} \right) x_k+\sum_{l\ne k}^K{\sum_{m=1}^M{\mathbf{v}_{mk}^{H}\mathbf{h}_{ml}x_l}}+\mathbf{n}_{k}^{\prime}.
\vspace{-0.3em}\end{equation}
Based on \eqref{decoding_signal}, we can compute the achievable SE for UE $k$ as $\mathrm{SE}_k=\frac{\tau _u}{\tau _c}\log _2(1+\mathrm{SINR}_k)$ with $\mathrm{SINR}_k$ being the effective signal-to-interference-plus-noise ratio (SINR) for UE $k$, which can be computed based on the UatF bound as shown in \eqref{SINR_k} at the top of this page. Notably, when the local decoding scheme in \eqref{local_decoding_scheme} is applied, we can derive the optimal weighting strategy $\left\{ \mathbf{W}_{mk}:\forall m\right\} $, which can maximize the achievable SE for UE $k$ as follows. 

\begin{prop}\label{prop_obe_weighting}
The optimal weighting matrices $\left\{ \mathbf{W}_{mk}:\forall m\right\} $, which maximize the achievable SE, can be constructed as $\mathbf{W}_{mk}=\mathrm{vec}^{-1}( \mathbf{w}_{mk})$ with $\mathbf{w}_{mk}=[ \mathbf{w}_{k} ] _{( m-1 ) N^2+1:mN^2}\in \mathbb{C} ^{N^2}$. The optimal collective vectorized weighting scheme for UE $k$ $\mathbf{w}_k=\left[ \mathbf{w}_{1k}^{T},\mathbf{w}_{2k}^{T},\dots ,\mathbf{w}_{Mk}^{T} \right] ^T\in \mathbb{C} ^{MN^2}$ is computed as
\setcounter{equation}{7}
\vspace{-0.55em}\begin{equation}\label{opt_weighting}
\begin{aligned}
\mathbf{w}_{k}^{*}=\left( \sum_{l=1}^K{p_l\mathbb{E} \{\mathbf{z}_{kl}\mathbf{z}_{kl}^{H}\}} +\sigma ^2\mathbf{\Theta }_k\right) ^{-1}\!\!\mathbb{E} \{\mathbf{g}_k\}
\end{aligned}
\vspace{-0.3em}\end{equation}
with the maximized SINR being
$\mathrm{SINR}_k^{*}=p_k\mathbb{E} \{\mathbf{g}_k\}^H( \sum_{l=1}^K{p_l\mathbb{E} \{\mathbf{z}_{kl}\mathbf{z}_{kl}^{H}\}-p_k\mathbb{E} \{\mathbf{g}_k\}\mathbb{E} \{\mathbf{g}_k\}^H}+\sigma ^2\mathbf{\Theta }_k ) ^{-1}\mathbb{E} \{\mathbf{g}_k\}$, where $\mathbf{z}_{kl}=[\mathrm{vec(}\mathbf{h}_{1l}\overline{\mathbf{v}}_{1k}^{H})^T,\dots ,\mathrm{vec(}\mathbf{h}_{Ml}\overline{\mathbf{v}}_{Mk}^{H})^T]\in \mathbb{C} ^{MN^2}$, $\mathbf{\Theta }_k=\mathrm{diag}\{(\mathbb{E} \{\overline{\mathbf{v}}_{1k}\overline{\mathbf{v}}_{1k}^{H}\}^T\otimes \mathbf{I}_N),\dots ,(\mathbb{E} \{\overline{\mathbf{v}}_{Mk}\overline{\mathbf{v}}_{Mk}^{H}\}^T\otimes \mathbf{I}_N)\}\in \mathbb{C} ^{MN^2\times MN^2}$, and $\mathbf{g}_k=[\mathrm{vec(}\mathbf{h}_{1k}\overline{\mathbf{v}}_{1k}^{H})^T,\dots ,\mathrm{vec(}\mathbf{h}_{Mk}\overline{\mathbf{v}}_{Mk}^{H})^T]^T\in \mathbb{C} ^{MN^2}$.
\end{prop}
\begin{IEEEproof}
The proof is provided in Appendix~\ref{app_obe_weighting}.
\end{IEEEproof}

\begin{rem}
Since this optimal weighting scheme follows from a similar idea to the OBE methodologies as in \cite{OBETrans,neumann2018bilinear}, we call it the OBE weighting scheme. The proposed OBE weighting scheme is a natural application of the conventional OBE methodology. Compared with \cite{OBETrans}, which applies the OBE principle to design the BE-structured combining schemes, that is $\overline{\mathbf{v}}_{mk}$ in the presented model as \eqref{local_decoding_scheme}, the present work instead applies it to optimize the weighting matrices $\mathbf{W}_{mk}$ for arbitrary local combining schemes $\overline{\mathbf{v}}_{mk}$. Moreover, $\mathbf{W}_{mk}$ is designed based on long-term statistical channel information rather than instantaneous CSI. Hence, it is a deterministic variable for each given set of scheduled users. Meanwhile, when $\mathbf{W}_{mk}=a_{mk}\mathbf{I}_N$, this matrix-valued OBE weighting scheme can reduce to the conventional scalar-valued LSFD weighting scheme, which can also be optimized as in \eqref{lsfd_eq}.
\end{rem}

The MSE of the final decoding data can be defined as $\mathrm{MSE}_k=\mathbb{E} \{ | x_k-\check{x}_k |^2  \}$, which is computed based on the channel statistics. It is insightful to observe that the OBE weighting scheme in \eqref{opt_weighting} can also minimize $\mathrm{MSE}_k$, which is proved in Appendix~\ref{MSE_Minimization}. Note that for arbitrary weighting schemes, we can construct $\mathrm{MSE}_k$ as
\vspace{-0.55em}\begin{equation}\label{MSE_k}
\mathrm{MSE}_k\!=\!\mathbb{E} \{ | x_k-\check{x}_k |^2 \} =\mathbf{w}_{k}^{H}\mathbf{\Gamma }_k\mathbf{w}_k-2p_k\Re ( \mathbf{w}_{k}^{H}\bm{\xi }_k ) +p_k,
\vspace{-0.3em}\end{equation}
where $\mathbf{\Gamma }_k=\sum_{l=1}^K{p_l\mathbb{E} \{ \mathbf{z}_{kl}\mathbf{z}_{kl}^{H} \}}+\sigma ^2\mathbf{\Theta }_k\in \mathbb{C} ^{MN^2\times MN^2}$ and $\bm{\xi }_k=\mathbb{E} \{ \mathbf{g}_k \} \in \mathbb{C} ^{MN^2}$. Then, we can also derive the total MSE for all $K$ UEs as
\vspace{-0.55em}\begin{equation}\label{total_MSE}
\sum_{k=1}^K{\mathrm{MSE}_k}=\mathbf{w}^H\mathbf{\Gamma w}-2\Re \left( \mathbf{w}^H\bm{\xi } \right) +\sum_{k=1}^K{p_k},
\vspace{-0.3em}\end{equation}
where $\mathbf{w}=[ \mathbf{w}_{1}^{T},\dots ,\mathbf{w}_{K}^{T} ] ^T\in \mathbb{C} ^{KMN^2}$ is the collective weighting vector for all UEs, $\bm{\xi }=[ p_1\bm{\xi }_{1}^{T},\dots ,p_K\bm{\xi }_{K}^{T} ] ^T\in \mathbb{C} ^{KMN^2}$, and $\mathbf{\Gamma }=\mathrm{diag}( \mathbf{\Gamma }_1,\dots \mathbf{\Gamma }_K ) \in \mathbb{C} ^{KMN^2\times KMN^2}$. We notice that finding the optimal $\mathbf{w}$, which minimizes the total MSE in \eqref{total_MSE}, is equivalent to finding the set of OBE weighting vectors $\{\mathbf{w}_{k}^{*}:\forall k\}$, that simultaneously minimize their respective MSEs in \eqref{MSE_k}. This is because $\mathrm{MSE}_k$ for different UEs are not affected by each other, which can be clearly showcased from \eqref{MSE_k}. These MSE-related formulations will be vital metrics in the following sparse optimization parts.

\vspace{-0.4cm}

\section{Power Consumption Model}\label{sec_power_model}
\vspace{-0.2cm}
To efficiently decrease the consumed power and enhance the energy efficiency while maintaining an acceptable data rate performance, it is important to utilize an accurate power consumption model, which can accurately depict the power consumption for different parts of systems. The power consumption model is also essential when we later develop sparsity optimization schemes, because these optimize different parts of the power consumption model. We provide a generalized power consumption model for the considered network that builds on \cite{10113887,bjornson2015optimal,8187178,ngo2017total}. It includes the following major components: 1) the power consumed at the array sites $\{ P_{m}^{\mathrm{ap}}: \forall m \} $; 2) the power consumed at the UEs $\{ P_{m}^{\mathrm{ue}}: \forall k \} $; 3) the power consumed at fronthaul connections between array sites and the CPU $\{P_{m}^{\mathrm{fh}}: \forall m \}$; and 4) the CPU power consumption $P^{\mathrm{cpu}}$. Thus, the total power consumption can be modeled as 
\vspace{-0.55em}\begin{equation}\label{total_power}
P=\sum_{k=1}^K{P_{k}^{\mathrm{ue}}}+\sum_{m=1}^M{P_{m}^{\mathrm{ap}}}+\sum_{m=1}^M{P_{m}^{\mathrm{fh}}}+P^{\mathrm{cpu}}.
\vspace{-0.3em}\end{equation}
We provide the modeling details below.

\textit{UE Power Consumption Model:} The consumed power at UE $k$ can be modeled as 
\vspace{-0.55em}\begin{equation}\label{power_ue}
P_{k}^{\mathrm{ue}}=P_{k}^{\mathrm{ue},\mathrm{c}}+\frac{\tau _pp_{k,p}+\tau _up_k}{\tau _c\eta _{\mathrm{ue}}},
\vspace{-0.3em}\end{equation}
where $P_{k}^{\mathrm{ue},\mathrm{c}}$ denotes the circuit power for UE $k$ and the second term denotes the consumed power of the uplink transmission with $0<\eta _{\mathrm{ue}}\leqslant 1$ being the power amplifier efficiency at UEs.

\textit{AP Power Consumption Model:} The power consumed at AP $m$ can be modeled as 
\vspace{-0.55em}\begin{equation}\label{power_ap}
P_{m}^{\mathrm{ap}}=B_mP_{m}^{\mathrm{array}}+\sum_{n=1}^N{A_{mn}P_{mn}^{\mathrm{ant},\mathrm{c}}}+\sum_{n=1}^N{\sum_{k=1}^K{a_{mnk}P_{mnk}^{\mathrm{pro}}}},
\vspace{-0.3em}\end{equation}
where $P_{m}^{\mathrm{array}}$ denotes the fixed array power for AP $m$ including the power for the site-cooling and control signaling \cite{bjornson2015optimal}, $P_{mn}^{\mathrm{ant},\mathrm{c}}$ denotes the internal circuit power for the $n$-th antenna of AP $m$, and $P_{mnk}^{\mathrm{pro}}$ represents the power consumption for the $n$-th antenna of AP $m$ for processing the received signal of UE $k$. $\{ B_m,A_{mn},a_{mnk} \} \in \{ 0,1 \} $ denote the binary symbols to describe the activating modes of AP $m$, the $n$-th antenna of AP $m$, and the $n$-th antenna of AP $m$ for UE $k$, respectively. Note that these symbols can be derived based on the value characteristics of $\{\mathbf{W}_{mk}: \forall m, \forall k\}$.
Specifically, $a_{mnk}=1$ if the $n$-th row of $\mathbf{W}_{mk}$ is nonzero, and $a_{mnk}=0$ otherwise. $A_{mn}=1$ if the $n$-th rows of $\mathbf{W}_{mk}$ are not all zero over all $k$, and $A_{mn}=0$ otherwise. Likewise, $B_m=1$ if at least one matrix $\mathbf{W}_{mk}$ is nonzero over all $k$, and $B_m=0$ otherwise.

\textit{Fronthaul Link Power Consumption Model:} The consumed power for the fronthaul link between AP $m$ and the CPU is modeled as 
\vspace{-0.55em}\begin{equation}\label{power_front}
P_{m}^{\mathrm{fh}}=B_mP_{m}^{\mathrm{fh},\mathrm{fix}}+\frac{\tau _u}{\tau _c} \sum_{k=1}^K{b_{mk}P_{mk}^{\mathrm{sig}}}  ,
\vspace{-0.3em}\end{equation}
where $P_{m}^{\mathrm{fh},\mathrm{fix}}$ is the fixed power consumed for the fronthaul link of AP $m$ and $P_{mk}^{\mathrm{sig}}$ denotes the required signaling power for UE $k$ for the fronthaul link of AP $m$.  
$b_{mk}\in \left\{ 0,1 \right\} $ denotes the activating mode of AP $m$ for UE $k$ and thus, $\sum_{k=1}^K{b_{mk}}$ denotes the number of UEs that AP $m$ serves.

\textit{CPU Power Consumption Model:} The power consumed at the CPU can be modeled as
\vspace{-0.55em}\begin{equation}\label{power_cpu}
P^{\mathrm{cpu}}=P^{\mathrm{cpu},\mathrm{fix}}+ \sum_{k=1}^K{\mathcal{B}\,\mathrm{SE}_kP^{\mathrm{cpu},\mathrm{dec}}},
\vspace{-0.3em}\end{equation}
where $P^{\mathrm{cpu},\mathrm{fix}}$ denotes the fixed power consumption at the CPU, $\mathcal{B} $ denotes the system bandwidth, and $P^{\mathrm{cpu},\mathrm{dec}}$ denotes the power consumption per bit for the decoding at the CPU.

Based on the achievable SE and the defined power consumption model, we can compute the total EE (in bit/Joule) for the studied network as \cite{10113887}
\vspace{-0.55em}\begin{equation}\label{ee}
\mathrm{EE}=  \sum_{k=1}^K \frac{\mathcal{B}\cdot\mathrm{SE}_k}{P}.
\vspace{-0.3em}\end{equation}

\vspace{-0.4cm}

\section{Antenna-Level Sparse Activation}\label{sec_antenna}
\vspace{-0.2cm}

To reduce the consumed power and enhance the EE, we can turn off some antenna/array-UE links based on the particular optimization goal. In this section, we focus on the antenna-level activation, where the sparse optimization methodologies are applied. More specifically, the antenna activation methodologies for one particular UE and for all UEs are proposed.

\vspace{-0.3cm}
\subsection{Problem Formulation}
Motivated by sparse reconstruction, we can formulate the generalized real-valued total MSE minimization problem 
\vspace{-0.55em}\begin{equation}\label{MSE_Sparse}
\underset{\underline{\bf{w}}\in \mathbb{R} ^{2KMN^2}}{\min}\underline{\bf{w}}^T\underline{\bf{\Gamma}}\underline{\bf{w}}-2\underline{\mathbf{w}}^T\underline{\bm{\xi}}+\Omega ( \underline{\mathbf{w}} ),
\vspace{-0.3em}\end{equation}
where $\underline{\mathbf{w}}=[ \underline{\mathbf{w}}_{1}^{T},\dots ,\underline{\mathbf{w}}_{K}^{T} ] ^T\in \mathbb{R} ^{2KMN^2}$, $\underline{\bm{\xi}}=[ p_1\underline{\bm{\xi}}_{1}^{T},\dots ,p_K\underline{\bm{\xi}}_{K}^{T} ] ^T\in \mathbb{R} ^{2KMN^2}$, and $\underline{\mathbf{\Gamma}}=\mathrm{diag}( \underline{\mathbf{\Gamma}}_1,\dots \underline{\mathbf{\Gamma}}_K ) \in \mathbb{R} ^{2KMN^2\times 2KMN^2}$ are all real-valued variables with 
\vspace{-0.55em}\begin{equation}\label{var_real}
\begin{aligned}
\underline{\mathbf{w}}_k &=
\left[ \begin{array}{c}
  \Re\left(\mathbf{w}_k\right) \\
  \Im\left(\mathbf{w}_k\right) \\
\end{array} \right] \in \mathbb{R}^{2MN^2},\quad
\underline{\boldsymbol{\xi}}_k =
\left[ \begin{array}{c}
  \Re\left(\boldsymbol{\xi}_k\right) \\
  \Im\left(\boldsymbol{\xi}_k\right) \\
\end{array} \right] \in \mathbb{R}^{2MN^2}
\\
\underline{\boldsymbol{\Gamma}}_k &=
\left[ \begin{array}{cc}
  \Re\left(\boldsymbol{\Gamma}_k\right) & -\Im\left(\boldsymbol{\Gamma}_k\right) \\
  \Im\left(\boldsymbol{\Gamma}_k\right) &  \Re\left(\boldsymbol{\Gamma}_k\right) \\
\end{array} \right] \in \mathbb{R}^{2MN^2\times 2MN^2}.
\end{aligned}
\vspace{-0.3em}\end{equation}
Note that $\underline{\bf{w}}^T\underline{\bf{\Gamma}}\underline{\bf{w}}-2\underline{\mathbf{w}}^T\underline{\bm{\xi}}$ is the MSE cost, which is a convex quadratic function of $\underline{\mathbf{w}}$. Meanwhile, $\Omega ( \underline{\mathbf{w}})$ represents an arbitrary sparsity-inducing function, which can encourage particular values in $\underline{\mathbf{w}}$ to be zero. By constructing different formulations of $\Omega ( \underline{\mathbf{w}})$, we can embrace different sparse optimization problems, where particular parts in $\underline{\mathbf{w}}$ will be forced to be $0$. Thus, we can formulate the antenna-level activation problem in this section and the array-level activation problem in the next section by constructing particular $\Omega ( \underline{\mathbf{w}})$ and solving corresponding MSE minimization problems.
\vspace*{-0.1cm}
\begin{rem}
The proposed sparse activation framework in the following operates at different granularities depending on the considered scheme. In the UE-specific antenna-level scheme in Sec.~\ref{sec_antenna_per_UE}, the activation state is defined for each antenna-UE serving link, such that a given antenna can be activated for some UEs and inactivated for others. In the network-wide antenna-level scheme for all UEs in Sec.~\ref{antenna_all_UE}, an antenna can be completely switched off over the whole UE set, while the remaining activated antennas are encouraged to serve only a limited number of UEs. In the UE-specific array-level scheme in Sec.~\ref{array_per_UE}, the activation state is defined for each array-UE serving link, meaning that all antennas within an array are jointly activated or inactivated for a particular UE. In the network-wide array-level scheme for all UEs in Sec.~\ref{array_all_UE}, some arrays can be entirely turned into sleep mode over the whole UE set, while the number of UEs served by each remaining activated array is controlled.    
\end{rem}

\vspace{-0.6cm}
\subsection{Antenna-Level Activation for Particular UE}\label{sec_antenna_per_UE}
In this part, we study the antenna activation for the particular UE, that is the UE-specific antenna activation, where the serving link between the antenna and the UE can be flexibly activated or not. The physical interpretation behind this antenna activation is to limit the average number of UEs that the array antennas provide service to. By doing this, the consumed signal processing power in \eqref{power_ap} as $\sum_{n=1}^N{\sum_{k=1}^K{a_{mnk}P_{mnk}^{\mathrm{pro}}}}$ can be effectively reduced by pushing some antenna-UE serving modes $a_{mnk}$ to be zero.\footnote{When introducing the characteristic of each sparse activation scheme, we only mention the directly affected parts of the power consumption model. The other power consumption parts may also have the potential to be reduced since the antenna/array activating modes in \eqref{total_power} are mutually coupled. For instance, if $\{a_{mnk}=0, \forall k\}$, we also have $A_{mn}=0$, and thus, the $n$-th antenna of AP $m$ will be turned off and the corresponding consumed power for this antenna will also be reduced, such as $A_{mn}P_{mn}^{\mathrm{ant},\mathrm{c}}$.}

We first take a look from the perspective of the weighted decoding scheme $\mathbf{v}_{mk}=\mathbf{W}_{mk}\overline{\mathbf{v}}_{mk}$. The $n$-th element of $\mathbf{v}_{mk}$ can be computed as $\mathbf{w}_{mkn}^{T}\overline{\mathbf{v}}_{mk}$, where $\mathbf{w}_{mkn}\in \mathbb{C} ^{N}$ denotes the $n$-th row of $\mathbf{W}_{mk}$ after transposition. We can observe that the activation mode between the $n$-th antenna for the $m$-th array and the $k$-th UE can be adjusted by $\mathbf{w}_{mkn}$, where $\mathbf{w}_{mkn}=\bf{0}$ denotes the scenario that the $n$-th antenna for the $m$-th array is inactivated for the $k$-th UE.
To facilitate the antenna activation for the particular UE, we define the sparsity-inducing function $\Omega ( \underline{\mathbf{w}})$ as
\vspace{-0.55em}\begin{equation}\label{sparse_term1}
\Omega \left( \underline{\mathbf{w}} \right) =\lambda \sum_{k=1}^K{\sum_{m=1}^M{\sum_{n=1}^N{\left\| \underline{\mathbf{w}}\left[ \underline{\mathcal{I}}_{mkn} \right] \right\| _2}}},
\vspace{-0.3em}\end{equation}
where $\lambda \geqslant 0$ is a tunable regularization factor, $\underline{\mathcal{I}}_{mkn}=\mathcal{I} _{mkn}\cup ( \mathcal{I} _{mkn}+KMN^2 ) $ represents the subset of element-level indexes of real-valued $n$-th row of $\mathbf{W}_{mk}$, that is $\underline{\mathbf{w}}_{mkn}=[ \Re ( \mathbf{w}_{mkn}^{T} ) ,\Im ( \mathbf{w}_{mkn}^{T} ) ] ^T\in \mathbb{R} ^{2N}$, with $\mathcal{I} _{mkn}=\{ ( k-1 ) MN^2+( m-1 ) N^2+n+( j-1 ) N: j=1,\dots ,N \} $ being the subset of element indexes of $\mathbf{w}_{mkn}$ in $\mathbf{w}$, and ``$+$" in $\underline{\mathcal{I}}_{mkn}$ denotes the operator that adds particular offset behind ``$+$" to all element indexs in the subset before ``$+$". One more intuitively equivalent formulation of \eqref{sparse_term1} is $\Omega \left( \underline{\mathbf{w}} \right) =\lambda \sum_{k=1}^K{\sum_{m=1}^M{\sum_{n=1}^N{\left\| \underline{\mathbf{w}}_{mkn} \right\| _2}}}$. By substituting \eqref{sparse_term1} into \eqref{MSE_Sparse}, we derive the MSE minimization problem
\vspace{-0.55em}\begin{equation}\label{sparse_problem_1}
\mathcal{P} ^{\mathrm{ant},( 1 )}\!\!:\!\!\underset{\underline{\mathbf{w}}\in \mathbb{R} ^{2KMN^2}}{\min}\!\underline{\mathbf{w}}^T\underline{\mathbf{\Gamma}}\underline{\mathbf{w}}-2\underline{\mathbf{w}}^T\underline{\bm{\xi}}\!+\!\lambda \!\sum_{k=1}^K\!{\sum_{m=1}^M\!{\sum_{n=1}^N{\!\| \underline{\mathbf{w}}[ \underline{\mathcal{I}}_{mkn} ] \| _2}}}.
\vspace{-0.3em}\end{equation}
Note that \eqref{sparse_problem_1} is a convex problem since \eqref{sparse_term1} is a standard $l_1/l_2$ group penalty, which is a convex function \cite{rish2014sparse}. Fortunately, we can divide the total MSE minimization problem in \eqref{sparse_problem_1} into $K$ subproblems since the MSEs for different UEs are not affected by each other and depend on different parts of $\underline{\mathbf{w}}$. Thus, we can formulate the $k$-th subproblem as
\vspace{-0.55em}\begin{equation}\label{sparse_problem_1_k}
\mathcal{P} _{k}^{\mathrm{ant},\left( 1 \right)}:\underset{\underline{\mathbf{w}}_k\in \mathbb{R} ^{2MN^2}}{\min}f\left( \underline{\mathbf{w}}_k \right) +\Omega \left( \underline{\mathbf{w}}_k \right),
\vspace{-0.3em}\end{equation}
where 
\vspace{-0.55em}\begin{equation}\label{f_k}
f\left( \underline{\mathbf{w}}_k \right) =\underline{\mathbf{w}}_{k}^{T}\underline{\mathbf{\Gamma}}_k\underline{\mathbf{w}}_k-2\underline{\mathbf{w}}_{k}^{T}\underline{\bm{\xi}}_k
\vspace{-0.3em}\end{equation} 
and $\Omega \left( \underline{\mathbf{w}}_k \right) =\lambda _k\sum_{m=1}^M{\sum_{n=1}^N{\left\| \underline{\mathbf{w}}_{mkn} \right\| _2}}$ with $\lambda _k$ being the regularization factor for the $k$-th subproblem. Similarly, we can derive the subsets of element indexes of $\mathbf{w}_{mkn}$ in $\mathbf{w}_{k}$ and $\underline{\mathbf{w}}_{mkn}$ in $\underline{\mathbf{w}}_{k}$ as $\mathcal{I} _{mkn}^{\prime}=\{ ( m-1 ) N^2+n+( j-1 ) N: j=1,\dots ,N \} $ and $\underline{\mathcal{I}}_{mkn}^{\prime}=\mathcal{I} _{mkn}^{\prime}\cup ( \mathcal{I} _{mkn}^{\prime}+MN^2 ) $, respectively. 

Since \eqref{sparse_problem_1_k} is convex with non-smooth sparsity-inducing penalties, we can apply the proximal methodologies to solve \eqref{sparse_problem_1_k} efficiently \cite{rish2014sparse}. More specifically, first, we can start the solution with the initial point $\underline{\mathbf{w}}_{k}^{0}$, which can be selected as the OBE weighting scheme in \eqref{opt_weighting}. Then, we can iteratively update $\underline{\mathbf{w}}_{k}^{i}$ with $i$ being the iteration index. Given with $\underline{\mathbf{w}}_{k}^{i}$ derived in the $i$-th iteration, the update of $\underline{\mathbf{w}}_{k}^{i+1}$  can be derived by solving the following proximal problem
\vspace{-0.55em}\begin{equation}\label{sparse_problem_1_k_proximal}
\underset{\underline{\mathbf{w}}_k\in \mathbb{R} ^{2MN^2}}{\min}\small{\frac{1}{2}}\left\| \underline{\mathbf{w}}_k-G_f\left( \underline{\mathbf{w}}_{k}^{i} \right) \right\| _{2}^{2}+\mu_k \Omega \left( \underline{\mathbf{w}}_k \right) ,
\vspace{-0.3em}\end{equation}
where $G_f\left( \underline{\mathbf{w}}_{k}^{i} \right) =\underline{\mathbf{w}}_{k}^{i}-\mu_k \nabla f\left( \underline{\mathbf{w}}_{k}^{i} \right) $ represents the gradient update with $\mu_k$ being the step length. For the optimization problem in \eqref{sparse_problem_1_k}, we have $\nabla f\left( \underline{\mathbf{w}}_k \right) =2\underline{\mathbf{\Gamma}}_k\underline{\mathbf{w}}_k-2\underline{\bm{\xi}}_k$ and
\vspace{-0.55em}\begin{equation}\label{sparse_problem_1_k_Gf}
G_f\left( \underline{\mathbf{w}}_{k}^{i} \right) =\underline{\mathbf{w}}_{k}^{i}-2\mu_k \underline{\mathbf{\Gamma}}_k\underline{\mathbf{w}}_{k}^{i}+2\mu_k \underline{\bm{\xi}}_k.
\vspace{-0.3em}\end{equation}
Note that the unique solution to \eqref{sparse_problem_1_k_proximal} can be derived as
\vspace{-0.55em}\begin{equation}\label{sparse_problem_1_k_Gf_prox}
\mathrm{Prox}_{\mu_k \lambda_k l_1/l_2}( G_f( \underline{\mathbf{w}}_{k}^{i} ) )\!=\!\!\underset{\underline{\mathbf{w}}_k\in \mathbb{R} ^{2MN^2}}{\mathrm{arg}\min}\small{\frac{1}{2}}\| \underline{\mathbf{w}}_k-G_f( \underline{\mathbf{w}}_{k}^{i} ) \| _{2}^{2}+\mu_k \Omega ( \underline{\mathbf{w}}_k ).
\vspace{-0.3em}\end{equation}
Thanks to the strong convexity, the closed-form solution for \eqref{sparse_problem_1_k_Gf} can be derived based on the following result from \cite[Sec. 5.3.3]{rish2014sparse}.

\begin{lemm}\label{lemma_prox}
For the standard $l_1/l_2$ group penalty in \eqref{sparse_problem_1_k}, the proximal operator can be computed in closed-form as 
\vspace{-0.55em}\begin{equation}\label{pro_operator}
\left[ \mathrm{Prox}_{\mu l_1/l_2}\left( \mathbf{x} \right) \right] _q=\left( 1-\frac{\mu\lambda}{\left\| \mathbf{x}_q \right\| _2} \right) _+\mathbf{x}_q,
\vspace{-0.3em}\end{equation}
where $\mathbf{x}_q$ denotes the projection of $\mathbf{x}$ on the $q$-th group with the penalty $\Omega ( \mathbf{x} ) =\lambda \sum_{q=1}^Q{\| \mathbf{x}_q \| _2}\,\,q=\{ 1,\dots ,Q \}$.
\end{lemm}

Thus, following the methodology in Lemma~\ref{lemma_prox}, we can also construct $\Omega ( \underline{\mathbf{w}}_k )$ in an intuitive group-structured formulation
\vspace{-0.55em}\begin{equation}
\Omega ( \underline{\mathbf{w}}_k ) =\lambda_k \sum_{q_1=1}^{Q_1}{\| \underline{\mathbf{w}}_k( \underline{\mathcal{I}}_{q_1k}^{\prime} ) \| _2}\,\,q_1=\left\{ 1,\dots ,Q_1 \right\} ,Q_1=MN,
\vspace{-0.3em}\end{equation}
where $\underline{\mathcal{I}}_{q_1k}^{\prime}$ is a concise representation of $\underline{\mathcal{I}}_{m( q_1 ) kn( q_1 )}^{\prime}=\mathcal{I} _{m( q_1 ) kn( q_1 )}^{\prime}\cup ( \mathcal{I} _{m( q_1 ) kn( q_1 )}^{\prime}+MN^2 ) $. Note that the $q_1$-th group in  $\underline{\mathbf{w}}_k$ denotes the vector between UE $k$ and the $n\left( q_1 \right) $-th antenna of the $m\left( q_1 \right)$-th array with $n\left( q_1 \right) =\mathrm{mod}\left( q_1-1,N \right) +1$ and $m\left( q_1 \right) =\lfloor (q_1-1)/{N} \rfloor +1$. Thus, relying on \eqref{pro_operator}, the $q_1$-th group-related elements in $\mathrm{Prox}_{\mu_k \lambda_k l_1/l_2}( G_f( \underline{\mathbf{w}}_{k}^{i} ) ) $ can be computed as
\vspace{-0.55em}\begin{equation}\label{pro_operator_Gf_q1}
\begin{aligned}
&\left[ \mathrm{Prox}_{\mu_k \lambda_k l_1/l_2}\left( G_f\left( \underline{\mathbf{w}}_{k}^{i} \right) \right) \right] _{q_1}\\
&=\left( 1-{\mu_k \lambda_k}/{\| [ G_f\left( \underline{\mathbf{w}}_{k}^{i} \right) ] _{\underline{\mathcal{I}}_{q_1k}^{\prime}} \| _2} \right) _+\left[ G_f\left( \underline{\mathbf{w}}_{k}^{i} \right) \right] _{\underline{\mathcal{I}}_{q_1k}^{\prime}},
\end{aligned}
\vspace{-0.3em}\end{equation}
where, for the convenience of understanding, we utilize $\left[ \mathrm{Prox}_{\mu_k \lambda_k l_1/l_2}\left( \cdot \right) \right] _{q_1}$ to denote the proximal operator for the projected vector for the $q_1$-th group and $\left[ \mathbf{a} \right] _{\underline{\mathcal{I}}_{q_1k}^{\prime}}$ to represent the intended projected vector by extracting the elements from $\mathbf{a}$  based on the element index set $\underline{\mathcal{I}}_{q_1k}^{\prime}$.

Based on the above proximal methodologies, we can update 
\vspace{-0.55em}\begin{equation}\label{update_w_1}
\underline{\mathbf{w}}_{k}^{i+1}\gets \mathrm{Prox}_{\mu_k \lambda_k l_1/l_2}( G_f( \widehat{\underline{\mathbf{w}}}_{k}^{i} ) )
\vspace{-0.3em}\end{equation}
with 
\vspace{-0.55em}\begin{equation}\label{update_w_hat_1}
\widehat{\underline{\mathbf{w}}}_{k}^{i}=\underline{\mathbf{w}}_{k}^{i}+\left( \frac{t_i-1}{t_{i+1}} \right) \left( \underline{\mathbf{w}}_{k}^{i}-\underline{\mathbf{w}}_{k}^{i-1} \right)
\vspace{-0.3em}\end{equation}
where 
\vspace{-0.55em}\begin{equation}\label{update_t_i_1}
t_{i+1}={(1+\sqrt{1+4t_{i}^{2}})}/{2}
\vspace{-0.3em}\end{equation}
is the FISTA step for accelerating the coverage \cite{rish2014sparse}. More specifically, based on \eqref{pro_operator_Gf_q1}, the elements in $\underline{\mathbf{w}}_{k}^{i+1}$ with the element indexes $\underline{\mathcal{I}}_{q_1k}^{\prime}$ can be updated as 
\vspace{-0.55em}\begin{equation}\label{update_w_1_q1}
\underline{\mathbf{w}}_{k}^{i+1}( \underline{\mathcal{I}}_{q_1k}^{\prime} ) \gets [ \mathrm{Prox}_{\mu_k \lambda_k l_1/l_2}( G_f( \widehat{\underline{\mathbf{w}}}_{k}^{i} ) ) ] _{q_1}.
\vspace{-0.3em}\end{equation}

The way to solve $\mathcal{P} ^{\mathrm{ant},( 1 )}$ in \eqref{sparse_problem_1} is summarized in Algorithm~\ref{algo_1}. To implement it efficiently, we adopt a warm-restart strategy for the regularization parameter $\lambda_k$ as in \cite[Algorithm 1]{10113887} to accelerate the convergence (also in the following all algorithms). Specifically, we start with a large regularization value $\lambda _{k,\mathrm{large}}$. Then, we iteratively shrink $\lambda _{k,\mathrm{large}}$ towards the target regularization factor $\lambda _k=\bar{\lambda}\lambda_{k,\max}$ with $\lambda_{k,\max}=2\| \boldsymbol{\xi }_k( \mathcal{I} _{q_1k}^{\prime} ) \| _2$ being the widely used reference regularization factor in the considered group $l_1/l_2$ norm \cite{tibshirani2012strong} and $\bar{\lambda}\in [ 0,1 ]$ being the relative regularization ratio. The sparse optimization problem with the particular regularization factor can be solved in turn. At each step, the solution of the previous subproblem is used as the initialization of the current one. For the choice of $\mu_k$, we can compute it as $\mu _k=1/\left( 2\delta _{\max}\left( \mathbf{\Gamma }_k \right) \right) $, where $2\delta _{\max}\left( \mathbf{\Gamma }_k \right)$ is an upper bound of the Lipschitz  constant of $\nabla f\left( \underline{\mathbf{w}}_k \right)$ \cite{bach2012optimization} with $\delta _{\max}\left( \mathbf{\Gamma }_k  \right) $ being the maximum eigenvalue for $\mathbf{\Gamma }_k$.

\begin{algorithm}[t!]
\label{algo_1}
%\doublespacing
\caption{Algorithm for solving $\mathcal{P} ^{\mathrm{ant},( 1 )}$}
\KwIn{$\underline{\mathbf{w}}^{*}$, $\underline{\bm{\Gamma}}$, $\underline{\bm{\xi}}$, group information $\{ q_1 \} $, index information $\{ \underline{\mathcal{I}}_{q_1k}^{\prime} \} $, $\{\mu_k\}$, $\{\lambda_k\}$, $\{\Omega ( \underline{\mathbf{w}}_k ) \} $, $i_{\max}$}

\KwOut{$\mathbf{w}\in \mathbb{C} ^{KMN^2}$}

\For{$k = 1,\ldots,K$}
{
{\bf Initiation:} $i=1$, $\widehat{\underline{\mathbf{w}}}_{k}^{0}\gets \underline{\mathbf{w}}_{k}^{*}$, $\underline{\mathbf{w}}_{k}^{0}\gets \underline{\mathbf{w}}_{k}^{*}$, $t_1=1$;\\

\Repeat{$i = i_{\max}$ or convergence}
{
Compute $G_f( \widehat{\underline{\mathbf{w}}}_{k}^{i-1} )$ based on \eqref{sparse_problem_1_k_Gf};\\
\For{$q_1 = 1,\ldots,Q_1$}
{
Compute $[ \mathrm{Prox}_{\mu \lambda_k l_1/l_2}( G_f( \widehat{\underline{\mathbf{w}}}_{k}^{i-1} ) ) ] _{q_1}$ based on \eqref{pro_operator_Gf_q1};\\
$\underline{\mathbf{w}}_{k}^{i}( \underline{\mathcal{I}}_{q_1k}^{\prime} ) \gets [ \mathrm{Prox}_{\mu \lambda_k l_1/l_2}( G_f( \widehat{\underline{\mathbf{w}}}_{k}^{i-1} ) ) ] _{q_1}$;}
Compute $t_{i+1}={(1+\sqrt{1+4t_{i}^{2}})}/{2}$;\\
Compute $\widehat{\underline{\mathbf{w}}}_{k}^{i}=\underline{\mathbf{w}}_{k}^{i}+\left( \frac{t_i-1}{t_{i+1}} \right) \left( \underline{\mathbf{w}}_{k}^{i}-\underline{\mathbf{w}}_{k}^{i-1} \right)$;\\
$i \leftarrow i + 1$;\\
}
}
Derive $\mathbf{w}\in \mathbb{C} ^{KMN^2}$ from the transformation of \eqref{var_real}.\\
\end{algorithm}

\vspace{-0.3cm}
\subsection{Antenna-Level Activation for All UEs}\label{antenna_all_UE}
In this part, we consider the antenna activation towards all UEs; that is, some of the antennas are turned into the sleep mode, not serving any UEs at all, while the remaining activated antennas also serve a limited number of UEs. This scheme is a network-wide antenna activation scheme for all UEs. The practical interpretation of this scheme is to limit the number of activated antennas and also limit the average number of UEs served by the remaining activated antennas. By doing so, the antenna circuit power $\sum_{n=1}^N{A_{mn}P_{mn}^{\mathrm{ant},\mathrm{c}}}$ and antenna signal processing power $\sum_{n=1}^N{\sum_{k=1}^K{a_{mnk}P_{mnk}^{\mathrm{pro}}}}$ in \eqref{power_ap} can be reduced by turning off some antennas based on the sparse activation scheme in this part. 

First, we  define the sparsity-inducing function as
\vspace{-0.55em}\begin{equation}\label{sparse_term2_1}
\Omega \left( \underline{\mathbf{w}} \right) =\gamma \sum_{m=1}^M{\sum_{n=1}^N{\left\| \underline{\mathbf{w}}\left[ \underline{\mathcal{J}}_{mn} \right] \right\| _2}},
\vspace{-0.3em}\end{equation}
where $\underline{\mathcal{J}}_{mn}=\underset{k=1}{\overset{K}{\cup}}\underline{\mathcal{I}}_{mkn}$ denotes the subset of element indexes of the collective real-valued weighting vector of the $n$-th antenna for the $m$-th array for all $K$ UEs $\underline{\mathbf{w}}_{mn}=[\Re (\mathbf{w}_{mn}^{T}),\Im (\mathbf{w}_{mn}^{T})]^T\in \mathbb{R} ^{2KN}$, where $\mathbf{w}_{mn}\in \mathbb{C} ^{KN}$ denotes the $n$-th row of $\mathbf{W}_m=\left[ \mathbf{W}_{m1},\dots ,\mathbf{W}_{mK} \right] \in \mathbb{C} ^{N\times KN}$ after transposition, and $\gamma $ denotes the regularization factor for \eqref{sparse_term2_1}. Moreover, it is useful to also include the sparsity penalty function as \eqref{sparse_term1} into \eqref{sparse_term2_1} to control the number of UEs served by the remaining activated antennas as\footnote{Note that, in the technical parts, we define the indexes of the array, antenna, and UE from the perspective of the group indexes, such as $q_2$ and $q_3$. One feasible implementation in the simulation parts is to traverse the indexes of the array, antenna, and UE and generate the respective group information.}
\vspace{-0.55em}\begin{equation}\label{sparse_term2}
\Omega ( \underline{\mathbf{w}} ) =\gamma \sum_{q_2=1}^{Q_2}{\| \underline{\mathbf{w}}( \underline{\mathcal{J}}_{q_2} ) \| _2}+\lambda \sum_{q_3=1}^{Q_3}{\| \underline{\mathbf{w}}( \underline{\mathcal{I}}_{q_3} ) \| _2},
\vspace{-0.3em}\end{equation}
where $\sum_{m=1}^M{\sum_{n=1}^N{\| \underline{\mathbf{w}}[ \underline{\mathcal{J} }_{mn} ] \| _2=}}\sum_{q_2=1}^{Q_2}{\| \underline{\mathbf{w}}( \underline{\mathcal{J}}_{m( q_2 ) n( q_2 )} ) \| _2}$ is an intuitive group representation with $Q_2=MN$, $m\left( q_2 \right) =\lfloor {(q_2-1)}/{N} \rfloor +1$, and $n( q_2 ) =\mathrm{mod}\left( q_2-1,N \right) +1$. Then, for convenience, we use $\underline{\mathcal{J}}_{q_2}$ to denote $\underline{\mathcal{J}}_{m\left( q_2 \right) n\left( q_2 \right)}=\underset{k=1}{\overset{K}{\cup}}\underline{\mathcal{I}}_{m\left( q_2 \right) kn\left( q_2 \right)}$. Moreover, the second part in \eqref{sparse_term2} is also a group representation, where $\underline{\mathcal{I}}_{q_3}$ is the simple representation for $\underline{\mathcal{I}}_{m( q_3 ) k( q_3 ) n( q_3 )}=\mathcal{I} _{m( q_3 ) k( q_3 ) n( q_3 )}\cup ( \mathcal{I} _{m( q_3 ) k( q_3 ) n( q_3 )}+KMN^2 ) $ with $\mathcal{I} _{m( q_3 ) k( q_3 ) n( q_3 )}=\{ ( k( q_3 ) -1 ) MN^2+( m( q_3 ) -1 ) N^2+n( q_3 ) +( j-1 ) N: j=1,\dots ,N \} $, $Q_3=MNK$, $m( q_3 ) =\lfloor {\mathrm{mod}\left( q_3-1,MN \right)}/{N} \rfloor +1$, $k\left( q_3 \right) =\lfloor {(q_3-1)}/{MN} \rfloor +1$, and $n( q_3 ) =\mathrm{mod}( q_3-1,N ) +1$. Thus, we can formulate the antenna activation problem for all UEs as
\vspace{-0.55em}\begin{equation}\label{sparse_problem_2}
\mathcal{P} ^{\mathrm{ant},\left( 2 \right)}\!\!:\!\!\!\!\underset{\underline{\mathbf{w}}\in \mathbb{R} ^{2KMN^2}}{\min}f( \underline{\mathbf{w}} ) +\gamma \!\!\sum_{q_2=1}^{Q_2}{\| \underline{\mathbf{w}}( \underline{\mathcal{J}}_{q_2} ) \| _2}+\lambda\!\! \sum_{q_3=1}^{Q_3}{\| \underline{\mathbf{w}}( \underline{\mathcal{I}}_{q_3} ) \| _2},
\vspace{-0.3em}\end{equation}
where 
\vspace{-0.55em}\begin{equation}\label{f}
f\left( \underline{\mathbf{w}} \right) =\underline{\mathbf{w}}^{T}\underline{\mathbf{\Gamma}}\underline{\mathbf{w}}-2\underline{\mathbf{w}}^{T}\underline{\bm{\xi}}.
\vspace{-0.3em}\end{equation}
Different from \eqref{sparse_problem_1}, the problem in \eqref{sparse_problem_2} cannot be solved in parallel because the sparsity-inducing function in \eqref{sparse_term2} involves coupling between UEs. It can be observed that $\Omega ( \underline{\mathbf{w}} )$ in \eqref{sparse_term2} has the standard tree-structured characteristics, where arbitrary two groups in \eqref{sparse_term2} are either disjoint or one is included in the other \cite{bach2012optimization}. The group $q_3$ is included in the group $q_2$, which can be clearly showcased from the definition of $\underline{\mathcal{J}}_{q_2}$, and thus, the group $q_3$ and $q_2$ are defined as the leaf group and root group, respectively. Meanwhile, two arbitrary groups in the leaf group set or the root group set are disjoint, respectively. Physically, each root group reflects the collective activity of one antenna across all UEs, so zeroing its corresponding variable block switches off that antenna for the whole network. Each leaf group reflects the activity of that antenna for one particular UE, so zeroing its corresponding variable block removes the associated antenna-UE serving link.

For the tree-structured group norm in \eqref{sparse_term2} with a more compact structure $\Omega : \mathbf{x}\mapsto \sum\nolimits_{g\in \mathcal{G}}^{}{\| \mathbf{x}_g \|}_2$, the proximal operator for the tree-structured $l_1/l_2$ group norm with the included order $g_1\preceq \dots \preceq g_m$ in the sets of groups $\mathcal{G} $ is $\mathrm{Prox}_{\mu \Omega}=\Pi _{g_m}\circ \cdots \circ \Pi _{g_1}$, where $\Pi _g$ represents the proximal operator $\mathbf{x}_g\mapsto \mathrm{prox}_{\mu l_2}( \mathbf{x}_g ) $ upon the subspace corresponding to group $g$ and we have $f\circ y( \cdot ) =f( y( \cdot ) ) $ for any function $f$ and $y$ \cite{bach2012optimization}. Thus, the vector for this tree-structured group in iteration $i+1$ can be updated following the order from the leaf groups $\{q_3\}$ to the root groups $\{q_2\}$ as
\vspace{-0.55em}\begin{equation}\label{update_w_2}
\underline{\mathbf{w}}^{i+1}\gets \mathrm{Prox}_{\mu \gamma l_1/l_2}\circ \mathrm{Prox}_{\mu \lambda l_1/l_2}( G_f( \widehat{\underline{\mathbf{w}}}^i ) ),
\vspace{-0.3em}\end{equation}
where 
\vspace{-0.55em}\begin{equation}\label{update_w_hat_2}
\widehat{\underline{\mathbf{w}}}^i=\underline{\mathbf{w}}^i+\left( \frac{t_i-1}{t_{i+1}} \right) \left( \underline{\mathbf{w}}^i-\underline{\mathbf{w}}^{i-1} \right)
\vspace{-0.3em}\end{equation}
with $t_{i+1}$ updated as in \eqref{update_t_i_1}. More specifically, first, among the leaf groups $\left\{ q_3 \right\} $, the elements in $\underline{\mathbf{w}}^{i+1}$ with the element index set $\underline{\mathcal{I}}_{q_3}$ can be updated as 
\vspace{-0.55em}\begin{equation}\label{update_w_leaf_2}
\underline{\mathbf{w}}^{i+1}( \underline{\mathcal{I}}_{q_3} ) \gets [ \mathrm{Prox}_{\mu \lambda l_1/l_2}( G_f( \widehat{\underline{\mathbf{w}}}^i ) ) ] _{q_3},
\vspace{-0.3em}\end{equation}
where
\vspace{-0.55em}\begin{equation}\label{pro_operator_Gf_q3}
\begin{aligned}
&[ \mathrm{Prox}_{\mu \lambda l_1/l_2}( G_f( \underline{\widehat{\mathbf{w}}}^{i} ) ) ] _{q_3}\\
&=( 1-{\mu \lambda}/{\| [ G_f( \underline{\widehat{\mathbf{w}}}^{i} ) ] _{\underline{\mathcal{I}}_{q_3}} \| _2} ) _+[ G_f( \underline{\widehat{\mathbf{w}}}^{i} ) ] _{\underline{\mathcal{I}}_{q_3}},
\end{aligned}
\vspace{-0.3em}\end{equation}
with the proximal operator defined similarly to \eqref{pro_operator_Gf_q1}. Then, among the root groups $\left\{ q_2 \right\} $, the elements in $\underline{\mathbf{w}}^{i+1}$ with the element index set $\underline{\mathcal{J}}_{q_2}$ can be further updated as
\vspace{-0.55em}\begin{equation}\label{update_w_root_2}
\underline{\mathbf{w}}^{i+1}( \underline{\mathcal{J}}_{q_2} ) \gets [ \mathrm{Prox}_{\mu \gamma l_1/l_2}( G_f( \widehat{\underline{\mathbf{w}}}^i ) ) ] _{q_2}
\vspace{-0.3em}\end{equation}
with 
\vspace{-0.55em}\begin{equation}\label{pro_operator_Gf_q2}
\begin{aligned}
&[ \mathrm{Prox}_{\mu \gamma l_1/l_2}( G_f( \underline{\widehat{\mathbf{w}}}^{i} ) ) ] _{q_2}\\
&=( 1-{\mu \gamma}/{\| [ G_f( \underline{\widehat{\mathbf{w}}}^{i} ) ] _{\underline{\mathcal{J}}_{q_2}} \| _2} ) _+[ G_f( \underline{\widehat{\mathbf{w}}}^{i} ) ] _{\underline{\mathcal{J}}_{q_2}}.
\end{aligned}
\vspace{-0.3em}\end{equation}
In summary, we present the algorithm to solve the sparse optimization problem $\mathcal{P} ^{\mathrm{ant},\left( 2 \right)}$ as \eqref{sparse_problem_2} in Algorithm~\ref{algo_2}. We model $\lambda=\bar{\lambda}\lambda_{\max}$ with $\lambda_{\max}=2\| \boldsymbol{\xi }( \mathcal{I} _{q_3}) \| _2$ and $\gamma=\bar{\gamma}\gamma_{\max}$ with $\gamma_{\max}=2\| \boldsymbol{\xi }( \mathcal{J} _{q_2}) \| _2$, where $\bar{\lambda}\in [ 0,1 ] $ and $\bar{\gamma}\in [ 0,1 ]$ denote the relative regularization ratios for the leaf groups and root groups, respectively. Meanwhile, we have $\mu =1/\left( 2\delta _{\max}\left( \mathbf{\Gamma } \right) \right) $.

\begin{algorithm}[t!]
\label{algo_2}
%\doublespacing
\caption{Algorithm for solving $\mathcal{P} ^{\mathrm{ant},( 2 )}$}
\KwIn{$\underline{\mathbf{w}}^{*}$, $\underline{\bm{\Gamma}}$, $\underline{\bm{\xi}}$, root group information $\{ q_2 \} $, leaf group information $\{ q_3 \} $, index information $\{ \underline{\mathcal{J}}_{q_2} \} $, index information $\{ \underline{\mathcal{I}}_{q_3} \} $, $\mu$, $\lambda$, $\gamma$, $\Omega ( \underline{\mathbf{w}})  $, $i_{\max}$}

\KwOut{$\mathbf{w}\in \mathbb{C} ^{KMN^2}$}

\For{$k = 1,\ldots,K$}
{
{\bf Initiation:} $i=1$, $\widehat{\underline{\mathbf{w}}}_{k}^{0}\gets \underline{\mathbf{w}}_{k}^{*}$, $\underline{\mathbf{w}}_{k}^{0}\gets \underline{\mathbf{w}}_{k}^{*}$, $t_1=1$;\\

\Repeat{$i = i_{\max}$ or convergence}
{
Compute $G_f( \widehat{\underline{\mathbf{w}}}_{k}^{i-1} )$ based on \eqref{sparse_problem_1_k_Gf};\\
\For{leaf groups $q_3 = 1,\ldots,Q_3$}
{
Compute $[ \mathrm{Prox}_{\mu \lambda l_1/l_2}( G_f( \widehat{\underline{\mathbf{w}}}_{k}^{i-1} ) ) ] _{q_3}$ based on \eqref{pro_operator_Gf_q3};\\
$\underline{\mathbf{w}}_{k}^{i}( \underline{\mathcal{I}}_{q_3}) \gets [ \mathrm{Prox}_{\mu \lambda l_1/l_2}( G_f( \widehat{\underline{\mathbf{w}}}_{k}^{i-1} ) ) ] _{q_3}$;}
\For{root groups $q_2 = 1,\ldots,Q_2$}
{
Compute $[ \mathrm{Prox}_{\mu \gamma l_1/l_2}( G_f( \widehat{\underline{\mathbf{w}}}_{k}^{i-1} ) ) ] _{q_2}$ based on \eqref{pro_operator_Gf_q2};\\
$\underline{\mathbf{w}}_{k}^{i}( \underline{\mathcal{J}}_{q_2}) \gets [ \mathrm{Prox}_{\mu \gamma l_1/l_2}( G_f( \widehat{\underline{\mathbf{w}}}_{k}^{i-1} ) ) ] _{q_2}$;}
Compute $t_{i+1}={(1+\sqrt{1+4t_{i}^{2}})}/{2}$;\\
Compute $\widehat{\underline{\mathbf{w}}}_{k}^{i}=\underline{\mathbf{w}}_{k}^{i}+\left( \frac{t_i-1}{t_{i+1}} \right) \left( \underline{\mathbf{w}}_{k}^{i}-\underline{\mathbf{w}}_{k}^{i-1} \right)$;\\
$i \leftarrow i + 1$;\\
}
}
Derive $\mathbf{w}\in \mathbb{C} ^{KMN^2}$ from the transformation of \eqref{var_real}.\\
\end{algorithm}

\vspace{-0.2cm}

\section{Array-Level Sparse Activation}\label{sec_array}
In Section~\ref{sec_antenna}, we studied antenna-level activations for UE-specific vision and for all UEs vision, which can lead to power consumption reductions. To further reduce the power consumption and enhance the EE with proper SE performance, we can reduce array-level power consumption by turning off some array-UE links or whole arrays. In this section, we study the array-level activation for the particular UE and for all UEs, that is, all antennas in a particular array are simultaneously activated or not for one particular UE and for all UEs.

\vspace{-0.3cm}
\subsection{Array-Level Activation for Particular UE}\label{array_per_UE}
We study the UE-specific array activation scheme, and the physical interpretation is to limit the average number of UEs served by the arrays. This scheme can be regarded as the typical AP-UE association problem in CF mMIMO networks \cite{10918608}. Relying on this, the antenna signal processing power $\sum_{n=1}^N{\sum_{k=1}^K{a_{mnk}P_{mnk}^{\mathrm{pro}}}}$ and signaling power in the fronthaul links ${\tau _u}/{\tau _c}( \sum\nolimits_{k=1}^K{b_{mk}P_{mk}^{\mathrm{sig}}})$ can be efficiently reduced by turning off some array-UE serving links. The array activation problem in this part can also be modeled as a sparse optimization problem, as in \eqref{MSE_Sparse}, by inducing the array activation penalty in $\Omega \left( \underline{\mathbf{w}} \right)$. For the array activation for the particular UE, we model the sparsity-inducing function as
\vspace{-0.55em}\begin{equation}\label{sparse_term3}
\Omega \left( \underline{\mathbf{w}} \right) =\rho \sum_{m=1}^M{\left\| \underline{\mathbf{w}}\left[ \underline{\mathcal{H}}_{mk}\right] \right\| _2},
\vspace{-0.3em}\end{equation}
where $\underline{\mathcal{H}}_{mk}=\mathcal{H} _{mk}\cup \left( \mathcal{H} _{mk}+KMN^2 \right)$ represents the subset of element indexes of $\underline{\mathbf{w}}_{mk}=[\Re (\mathbf{w}_{mk}^{T}),\Im (\mathbf{w}_{mk}^{T})]^T\in \mathbb{R} ^{2N^2}$ in $\underline{\mathbf{w}}$ with $\mathcal{H} _{mk}=\{ ( ( k-1 ) MN^2+( m-1 ) N^2+1 ) :( ( k-1 ) MN^2+mN^2 ) \} $ being the element indexes of $\mathbf{w}_{mk}$ in $\mathbf{w}$. Thus, based on \eqref{sparse_term3}, we can formulate the array activation problem for the particular UE as 
\vspace{-0.55em}\begin{equation}\label{sparse_problem_3}
\mathcal{P} ^{\mathrm{array},\left( 1 \right)}:\underset{\underline{\mathbf{w}}\in \mathbb{R} ^{2KMN^2}}{\min}f( \underline{\mathbf{w}} ) +\rho \sum_{m=1}^M{\left\| \underline{\mathbf{w}}\left[ \underline{\mathcal{H}}_{mk}\right] \right\| _2},
\vspace{-0.3em}\end{equation}
where $f( \underline{\mathbf{w}} )$ is given as in \eqref{f} and $\rho$ is the regularization factor.
Fortunately, $\mathcal{P} ^{\mathrm{array},\left( 1 \right)}$ can also be solved in parallel to all $K$ UEs, where the $k$-th subproblem can be formulated as
\vspace{-0.55em}\begin{equation}\label{sparse_problem_3_k}
\mathcal{P}_k ^{\mathrm{array},\left( 1 \right)}:\underset{\underline{\mathbf{w}}_k\in \mathbb{R} ^{2MN^2}}{\min}f\left( \underline{\mathbf{w}}_k \right) +\rho_k \sum_{m=1}^M{\left\| \underline{\mathbf{w}}_k\left( \underline{\mathcal{H}}_{mk}^{\prime} \right) \right\| _2},
\vspace{-0.3em}\end{equation}
where $f\left( \underline{\mathbf{w}}_k \right)$ is shown in \eqref{f_k}, $\rho_k$ denotes the regularization factor for the $k$-th subproblem, and $\underline{\mathcal{H}}_{mk}^{\prime}=\mathcal{H} _{mk}^{\prime}\cup \left( \mathcal{H} _{mk}^{\prime}+MN^2 \right) $ represents the subset of element indexes of $\underline{\mathbf{w}}_{mk}$ in $\underline{\mathbf{w}}_{k}$ with $\mathcal{H} _{mk}^{\prime}=\left\{ \left( \left( m-1 \right) N^2+1 \right) :mN^2 \right\} $ being the subset of element indexes of ${\mathbf{w}}_{mk}$ in ${\mathbf{w}}_{k}$. Focusing on \eqref{sparse_problem_3_k}, we can further denote $\Omega \left( \underline{\mathbf{w}}_k \right)=\rho_k \sum_{m=1}^M{\left\| \underline{\mathbf{w}}_k\left( \underline{\mathcal{H}}_{mk}^{\prime} \right) \right\| _2}$ as the intuitive group form\footnote{The involvement of $q_4$ is to simply replace $m$ in $\underline{\mathcal{H}}_{mk}^{\prime}$ as $q_4$. But we still utilize these $q$-related symbols to describe the group characteristics following the similar style as other sparse optimization problems in this paper.} $\Omega ( \underline{\mathbf{w}}_k ) =\rho _k\sum_{q_4=1}^{Q_4}{\| \underline{\mathbf{w}}_{k}(\underline{\mathcal{H}}_{q_4k}^{\prime}) \| _2}$ with $\underline{\mathcal{H}}_{q_4k}^{\prime}=\mathcal{H} _{q_4k}^{\prime}\cup ( \mathcal{H} _{q_4k}^{\prime}+MN^2 ) $ and  $q_4=\{ 1,\dots ,Q_4 \} ,Q_4=M$. We can observe that $\Omega \left( \underline{\mathbf{w}}_k \right)$ in this part is also a standard $l_1/l_2$ norm, and thus, we can derive the proximal operator as 
\vspace{-0.55em}\begin{equation}\label{pro_operator_Gf_q4}
\begin{aligned}
&[ \mathrm{Prox}_{\mu_k \rho_k l_1/l_2}( G_f(  \widehat{\underline{\mathbf{w}}}_{k}^{i}) ) ] _{q_4}\\
&=( 1-{\mu_k \rho_k}/{\| [ G_f(  \widehat{\underline{\mathbf{w}}}_{k}^{i} ) ] _{\underline{\mathcal{H}}_{q_4k}^{\prime}} \| _2} ) _+[ G_f(  \widehat{\underline{\mathbf{w}}}_{k}^{i} ) ] _{\underline{\mathcal{H}}_{q_4k}^{\prime}}.
\end{aligned}
\vspace{-0.3em}\end{equation}
We can easily derive the update steps for the sparse weighting vector in this part based on the methodology in Section~\ref{sec_antenna_per_UE} by simply replacing the group information $\{q_1\}$ and index information $\{ \underline{\mathcal{I}}_{q_1k}^{\prime} \} $ in Section~\ref{sec_antenna_per_UE} as the group information $\{q_4\}$ and index information $\{ \underline{\mathcal{H}}_{q_4k}^{\prime} \} $ in this part, respectively. And thus, we can derive the algorithm to solve $\mathcal{P} ^{\mathrm{array},\left( 1 \right)}$ in \eqref{sparse_problem_3} as Algorithm~\ref{algo_3}, where we have $\rho_k=\bar{\rho}\rho_{k,\max}$ with $\rho_{k,\max}=2\| \boldsymbol{\xi }_k( \mathcal{H} _{q_4k}^{\prime} ) \| _2$ and $\bar{\rho}\in [ 0,1 ]$, and $\mu _k=1/\left( 2\delta _{\max}\left( \mathbf{\Gamma }_k \right) \right) $.

\begin{algorithm}[t!]
\label{algo_3}
%\doublespacing
\caption{Algorithm for solving $\mathcal{P} ^{\mathrm{array},\left( 1 \right)}$}
\KwIn{$\underline{\mathbf{w}}^{*}$, $\underline{\bm{\Gamma}}$, $\underline{\bm{\xi}}$, group information $\{ q_4 \} $, index information $\{ \underline{\mathcal{I}}_{q_1k}^{\prime} \} $, $\{\mu_k\}$, $\{\lambda_k\}$, $\{\Omega ( \underline{\mathbf{w}}_k ) \} $, $i_{\max}$}

\KwOut{$\mathbf{w}\in \mathbb{C} ^{KMN^2}$}

\For{$k = 1,\ldots,K$}
{
{\bf Initiation:} $i=1$, $\widehat{\underline{\mathbf{w}}}_{k}^{0}\gets \underline{\mathbf{w}}_{k}^{*}$, $\underline{\mathbf{w}}_{k}^{0}\gets \underline{\mathbf{w}}_{k}^{*}$, $t_1=1$;\\

\Repeat{$i = i_{\max}$ or convergence}
{
Compute $G_f( \widehat{\underline{\mathbf{w}}}_{k}^{i-1} )$ based on \eqref{sparse_problem_1_k_Gf};\\
\For{$q_4 = 1,\ldots,Q_4$}
{
Compute $[ \mathrm{Prox}_{\mu_k \rho_k l_1/l_2}( G_f(  \widehat{\underline{\mathbf{w}}}_{k}^{i}) ) ] _{q_4}$ based on \eqref{pro_operator_Gf_q4};\\
$\underline{\mathbf{w}}_{k}^{i}( \underline{\mathcal{H}}_{q_4k}^{\prime} ) \gets [ \mathrm{Prox}_{\mu \rho_k l_1/l_2}( G_f( \widehat{\underline{\mathbf{w}}}_{k}^{i-1} ) ) ] _{q_4}$;}
Compute $t_{i+1}={(1+\sqrt{1+4t_{i}^{2}})}/{2}$;\\
Compute $\widehat{\underline{\mathbf{w}}}_{k}^{i}=\underline{\mathbf{w}}_{k}^{i}+\left( \frac{t_i-1}{t_{i+1}} \right) \left( \underline{\mathbf{w}}_{k}^{i}-\underline{\mathbf{w}}_{k}^{i-1} \right)$;\\
$i \leftarrow i + 1$;\\
}
}
Derive $\mathbf{w}\in \mathbb{C} ^{KMN^2}$ from the transformation of \eqref{var_real}.\\
\end{algorithm}

\begin{algorithm}[t!]
\label{algo_4}
%\doublespacing
\caption{Algorithm for solving $\mathcal{P} ^{\mathrm{array},\left( 2 \right)}$}
\KwIn{$\underline{\mathbf{w}}^{*}$, $\underline{\bm{\Gamma}}$, $\underline{\bm{\xi}}$, root group information $\{ q_5 \} $, leaf group information $\{ q_6 \} $, index information $\{ \underline{\mathcal{F}}_{q_5} \} $, index information $\{ \underline{\mathcal{H}}_{q_6} \} $, $\mu$, $\omega$, $\rho$, $\Omega ( \underline{\mathbf{w}}) $, $i_{\max}$}

\KwOut{$\mathbf{w}\in \mathbb{C} ^{KMN^2}$}

\For{$k = 1,\ldots,K$}
{
{\bf Initiation:} $i=1$, $\widehat{\underline{\mathbf{w}}}_{k}^{0}\gets \underline{\mathbf{w}}_{k}^{*}$, $\underline{\mathbf{w}}_{k}^{0}\gets \underline{\mathbf{w}}_{k}^{*}$, $t_1=1$;\\

\Repeat{$i = i_{\max}$ or convergence}
{
Compute $G_f( \widehat{\underline{\mathbf{w}}}_{k}^{i-1} )$ based on \eqref{sparse_problem_1_k_Gf};\\
\For{leaf groups $q_6 = 1,\ldots,Q_6$}
{
Compute $[ \mathrm{Prox}_{\mu \rho l_1/l_2}( G_f( \widehat{\underline{\mathbf{w}}}_{k}^{i-1} ) ) ] _{q_6}$ based on \eqref{pro_operator_Gf_q6};\\
$\underline{\mathbf{w}}_{k}^{i}( \underline{\mathcal{H}}_{q_6}) \gets [ \mathrm{Prox}_{\mu \rho l_1/l_2}( G_f( \widehat{\underline{\mathbf{w}}}_{k}^{i-1} ) ) ] _{q_6}$;}
\For{root groups $q_5 = 1,\ldots,Q_5$}
{
Compute $[ \mathrm{Prox}_{\mu \omega l_1/l_2}( G_f( \widehat{\underline{\mathbf{w}}}_{k}^{i-1} ) ) ] _{q_5}$ based on \eqref{pro_operator_Gf_q5};\\
$\underline{\mathbf{w}}_{k}^{i}( \underline{\mathcal{F}}_{q_5}) \gets [ \mathrm{Prox}_{\mu \omega l_1/l_2}( G_f( \widehat{\underline{\mathbf{w}}}_{k}^{i-1} ) ) ] _{q_5}$;}
Compute $t_{i+1}={(1+\sqrt{1+4t_{i}^{2}})}/{2}$;\\
Compute $\widehat{\underline{\mathbf{w}}}_{k}^{i}=\underline{\mathbf{w}}_{k}^{i}+\left( \frac{t_i-1}{t_{i+1}} \right) \left( \underline{\mathbf{w}}_{k}^{i}-\underline{\mathbf{w}}_{k}^{i-1} \right)$;\\
$i \leftarrow i + 1$;\\
}
}
Derive $\mathbf{w}\in \mathbb{C} ^{KMN^2}$ from the transformation of \eqref{var_real}.\\
\end{algorithm}

\vspace{-0.3cm}
\subsection{Array-Level Activation for All UEs}\label{array_all_UE}
We will now study the array activation for all UEs, where some arrays are turned into the sleep mode, not serving any UE at all, while the number of UEs served by the remaining activated arrays is controlled. The physical interpretation behind this activation scheme is to not only limit the number of activated arrays, but also limit the average number of UEs served by the activated arrays. By doing so, many parts in the power consumption model can be significantly reduced, such as all parts in the consumed power at AP $P_{m}^{\mathrm{ap}}$ and all parts in the consumed power for the fronthaul link of AP $P_{m}^{\mathrm{fh}}$, by turning particular AP into the sleep mode.

We define the collective weighting vector between the $m$-th array and all $K$ UEs as $\mathbf{w}_m=[ \mathbf{w}_{m1}^{T},\dots ,\mathbf{w}_{mK}^{T} ] ^T\in \mathbb{C} ^{KN^2}$. Thus, we can derive the subsets of the element indexes of $\mathbf{w}_m$ in  $\mathbf{w}$ and $\underline{\mathbf{w}}_{m}=[\Re (\mathbf{w}_{m}^{T}),\Im (\mathbf{w}_{m}^{T})]^T\in \mathbb{R} ^{2KN^2}$ in $\underline{\mathbf{w}}$ as $\mathcal{F} _m=\underset{k=1}{\overset{K}{\cup}}\mathcal{H} _{mk}$ and $\underline{\mathcal{F}}_m=\underset{k=1}{\overset{K}{\cup}}\underline{\mathcal{H}}_{mk}$, respectively. Then, we can define the sparsity-inducing function, which can encourage some arrays' collective weighting vectors to be $\bf{0}$, as 
\vspace{-0.55em}\begin{equation}\label{sparse_term41}
\Omega ( \underline{\mathbf{w}} ) =\omega \sum_{m=1}^M{\| \underline{\mathbf{w}}[ \underline{\mathcal{F}}_m ] \| _2},
\vspace{-0.3em}\end{equation}
where $\omega$ is the regularization factor. We can also involve an additional sparsity-inducing term to limit the number of UEs served by the activated arrays: 
\vspace{-0.55em}\begin{equation}\label{sparse_term4}
\Omega ( \underline{\mathbf{w}} ) =\omega \sum_{q_5=1}^{Q_5}{\| \underline{\mathbf{w}}[ \underline{\mathcal{F}}_{q_5} ] \| _2}+\rho \sum_{q_6=1}^{Q_6}{\| \underline{\mathbf{w}}( \underline{\mathcal{H}}_{q_6} ) \| _2},
\vspace{-0.3em}\end{equation}
where the first term is the $q$-series group representation of \eqref{sparse_term41} with $Q_5=M$ and the second term is derived from $\rho \sum_{k=1}^K{\sum_{m=1}^M{\| \underline{\mathbf{w}}[ \underline{\mathcal{H}}_{mk} ] \| _2}}=\rho \sum_{q_6=1}^{Q_6}{\| \underline{\mathbf{w}}( \underline{\mathcal{H}}_{q_6} ) \| _2}$ with $Q_6=MN$ and $\underline{\mathcal{H}}_{q_6}$ being the concise representation of $\mathcal{H} _{m( q_6 ) k( q_6 )}=\{ ( ( k( q_6 ) -1 ) MN^2+( m\left( q_6 ) -1 ) N^2+1 ) :( ( k( q_6 ) -1 ) MN^2+m( q_6 \right) N^2 ) \} $, where $m( q_6 ) =\mathrm{mod}( q_6-1,M ) +1$ and $k( q_6 ) =\lfloor {(q_6-1)}/{M} \rfloor +1$. Furthermore, we can formulate the sparse optimization problem for the array activation for all UEs as
\vspace{-0.55em}\begin{equation}\label{sparse_problem_4}
\mathcal{P} ^{\mathrm{array},\left( 2 \right)}\!\!:\!\!\!\!\!\!\underset{\underline{\mathbf{w}}\in \mathbb{R} ^{2KMN^2}}{\min}f( \underline{\mathbf{w}} ) +\omega \!\!\sum_{q_5=1}^{Q_5}{\!\!\| \underline{\mathbf{w}}( \underline{\mathcal{F}}_{q_5} ) \| _2}+\rho\!\! \sum_{q_6=1}^{Q_6}{\!\!\| \underline{\mathbf{w}}( \underline{\mathcal{H}}_{q_6} ) \| _2},
\vspace{-0.3em}\end{equation}
where $f( \underline{\mathbf{w}} )$ is given in \eqref{f}. Note that $\mathcal{P} ^{\mathrm{array},\left( 2 \right)}$ in \eqref{sparse_problem_4} cannot be solved in parallel, thus we focus on directly solving \eqref{sparse_problem_4} as follows. It can be observed that $\Omega ( \underline{\mathbf{w}} ) $ in \eqref{sparse_term4} also embraces the tree-structured characteristic, where $\{q_5\}$ and $\{q_6\}$ are the root groups and leaf groups, respectively, since $\underline{\mathcal{H}}_{q_6}$ is included in $\underline{\mathcal{F}}_{q_5}$. Following the methodology in Section~\ref{antenna_all_UE}, we can derive the update of this tree-structured $l_1/l_2$ group as 
\vspace{-0.55em}\begin{equation}\label{update_w_4}
\underline{\mathbf{w}}^{i+1}\gets \mathrm{Prox}_{\mu \omega l_1/l_2}\circ \mathrm{Prox}_{\mu \rho l_1/l_2}( G_f( \widehat{\underline{\mathbf{w}}}^i ) ),
\vspace{-0.3em}\end{equation}
where $\widehat{\underline{\mathbf{w}}}^i$ is computed as \eqref{update_w_hat_2} with $t_{i+1}$ updated as \eqref{update_t_i_1}. First, among the leaf groups $\left\{ q_6 \right\} $, we can update the elements in $\underline{\mathbf{w}}^{i+1}$ with the element index set $\underline{\mathcal{H}}_{q_6}$ as
\vspace{-0.55em}\begin{equation}\label{update_w_leaf_4}
\underline{\mathbf{w}}^{i+1}( \underline{\mathcal{H}}_{q_6} ) \gets [ \mathrm{Prox}_{\mu \rho l_1/l_2}( G_f( \widehat{\underline{\mathbf{w}}}^i ) ) ] _{q_6},
\vspace{-0.3em}\end{equation}
where 
\vspace{-0.55em}\begin{equation}\label{pro_operator_Gf_q6}
\begin{aligned}
&[ \mathrm{Prox}_{\mu \rho l_1/l_2}( G_f( \underline{\widehat{\mathbf{w}}}^{i} ) ) ] _{q_6}\\
&=( 1-{\mu \rho}/{\| [ G_f( \underline{\widehat{\mathbf{w}}}^{i} ) ] _{\underline{\mathcal{H}}_{q_6}} \| _2} ) _+[ G_f( \underline{\widehat{\mathbf{w}}}^{i} ) ] _{\underline{\mathcal{H}}_{q_6}}.
\end{aligned}
\vspace{-0.3em}\end{equation}
Then, for the root groups $\left\{ q_5 \right\} $, we can update the elements in $\underline{\mathbf{w}}^{i+1}$ with the element index set $\underline{\mathcal{F}}_{q_5}$ as 
\vspace{-0.55em}\begin{equation}\label{update_w_root_4}
\underline{\mathbf{w}}^{i+1}( \underline{\mathcal{F}}_{q_5} ) \gets [ \mathrm{Prox}_{\mu \omega l_1/l_2}( G_f( \widehat{\underline{\mathbf{w}}}^i ) ) ] _{q_5}, 
\vspace{-0.3em}\end{equation}
where 
\vspace{-0.55em}\begin{equation}\label{pro_operator_Gf_q5}
\begin{aligned}
&[ \mathrm{Prox}_{\mu \omega l_1/l_2}( G_f( \underline{\widehat{\mathbf{w}}}^{i} ) ) ] _{q_5}\\
&=( 1-{\mu \omega}/{\| [ G_f( \underline{\widehat{\mathbf{w}}}^{i} ) ] _{\underline{\mathcal{F}}_{q_5}} \| _2} ) _+[ G_f( \underline{\widehat{\mathbf{w}}}^{i} ) ] _{\underline{\mathcal{F}}_{q_5}}.
\end{aligned}
\vspace{-0.3em}\end{equation}
The other update steps can be easily derived based on the approach in Section~\ref{antenna_all_UE} by replacing the group characteristics in Algorithm~\ref{algo_2} with the group characteristics in this part. Thus, we present the algorithm for solving $\mathcal{P} ^{\mathrm{array},\left( 2 \right)}$ in \eqref{sparse_problem_4} as Algorithm~\ref{algo_4}, where $\omega=\bar{\omega}\omega_{\max}$ with $\omega_{\max}=2\| \boldsymbol{\xi }( \mathcal{F} _{q_5}) \| _2$, $\rho=\bar{\rho}\rho_{\max}$ with $\rho_{\max}=2\| \boldsymbol{\xi }( \mathcal{H} _{q_6}) \| _2$, and $\bar{\omega}\in [ 0,1 ] $ and $\bar{\rho}\in [ 0,1 ]$ denote the relative regularization ratios for the leaf and root groups, respectively. We have $\mu =1/\left( 2\delta _{\max}\left( \mathbf{\Gamma } \right) \right) $.

\vspace{-0.1cm}
\section{Practical Implementation, Complexity, Convergence, and Fronthaul Analysis}
In this section, we deliver discussions about the practical implementation aspects, computational complexity, and convergence analysis for the studied schemes.

\vspace{-0.4cm}
\subsection{Practical Implementation}
In this part, we discuss the proposed schemes from the perspective of scalability, implementation feasibility, and industry discussions as in the following remarks.
\begin{rem}
The proposed sparse activation methods are motivated by the user-centric selective-service principle in scalable CF mMIMO \cite{9064545}, where each UE is served by a subset of APs rather than by all APs. However, the methods developed in this paper are not claimed to constitute strictly scalable implementations under the definition in \cite{9064545}, which requires the per-AP signal processing, fronthaul signaling, and associated resource requirements to remain finite as the number of UEs grows to infinity. 
\end{rem}

\vspace{-0.3cm}
\begin{rem}
The proposed antenna/array activation should be interpreted as a slow-timescale network reconfiguration policy, instead of a rapid per-slot hardware switching mechanism. Since the sparse activation design is built upon the long-term statistics-based weighting matrices $\mathbf{W}_{mk}$, the activation pattern only needs to be updated when the large-scale channel statistics change noticeably, rather than at every coherence block. Therefore, in practice, the sparse optimization can be implemented at the CPU based on long-term statistical information, while the resulting activation pattern is reused over many coherence blocks.
\end{rem}
\vspace{-0.3cm}

\begin{rem}
The considered sparse activation problem is relevant not only from a theoretical viewpoint, but also from practical standardization and industrial perspectives. Public third generation partnership project (3GPP) materials highlight network energy saving as an important direction in new radio, with antenna muting being one representative mechanism\footnote{https://www.3gpp.org/technologies/deep-dive/ee-article}. Industry has also emphasized energy-efficient radio operation through sleep modes and antenna-port adaptation in radio access networks\footnote{https://www.ericsson.com/en/reports-and-papers/white-papers/5g-advanced-evolution-towards-6g}. Nevertheless, antenna-level inactivation within an activated array should be viewed as a fine-grained energy-saving mechanism, rather than as a substitute for whole-array sleep, since the array-level fixed-power terms may remain active and may play a more important role in the overall system energy consumption and EE, as will be illustrated by numerical results.
\end{rem}

\vspace{-0.7cm}
\subsection{Computational Complexity}
In this part, we deliver the computational complexity analysis for the proposed schemes, including the OBE weighting scheme and four sparse activation schemes.
For the OBE weighting scheme in Proposition~\ref{prop_obe_weighting}, its dominant statistical precomputation cost is $\mathcal{O}(M^2K^2N^4I_r)$, where $I_r$ denotes the number of random realizations for the approximation of the required expectations by Monte-Carlo averaging. Besides, the subsequent weighting vector computational complexity is $\mathcal{O}(M^3KN^6)$. Therefore, the total computational complexity of the OBE weighting is $\mathcal{O}(M^2K^2N^4I_r + M^3KN^6)$. For the LSFD weighting scheme in \eqref{lsfd_eq}, the dominant statistical precomputation cost comes from constructing $\bm{\Xi}_{kl}$, which requires $\mathcal{O}(M^2K^2I_r)$. And then, the LSFD weighting vector is obtained with an additional complexity of $\mathcal{O}(M^3K)$. Therefore, the total complexity of LSFD is $\mathcal{O}(M^2K^2I_r + M^3K)$.

\vspace{-0.1cm}
\begin{rem}
Although the OBE weighting incurs a higher computational complexity than the LSFD weighting, this additional burden remains manageable in practice, since it is confined to slow-timescale updates based on long-term statistics rather than per-coherence-block online processing. More importantly, in contrast to the scalar LSFD weighting, the matrix-valued OBE weighting structure can exploit the antenna-domain degrees of freedom more effectively and explicitly shape the contribution of individual antenna elements, thereby enabling finer-grained weighting and improved performance.
\end{rem}

For the proposed four sparse activation algorithms, the per-iteration complexity mainly consists of the gradient update, the proximal update, and the FISTA extrapolation. The following complexity analysis only counts the iterative computational cost of Algorithms~\ref{algo_1} to Algorithms~\ref{algo_4}, while assuming that the long-term statistical quantities $\bm{\Gamma}$ and $\bm{\xi}$ have been precomputed. For the UE-specific schemes, i.e., Algorithm~\ref{algo_1} and Algorithm~\ref{algo_3}, the dominant cost in each iteration comes from the gradient computation involving the matrix-vector multiplication with $\bm{\Gamma}_k$, whose complexity is $\mathcal{O}(M^2N^4)$. The proximal update and the FISTA extrapolation only require group-wise operations and vector updates, both having a lower complexity order $\mathcal{O}(MN^2)$. Hence, the per-iteration complexity of each UE-specific subproblem is $ \mathcal{O}(M^2N^4+MN^2)=\mathcal{O}(M^2N^4)$, which is dominated by the gradient step. If all $K$ UE subproblems are solved serially, the overall per-iteration complexity becomes $ \mathcal{O}(M^2KN^4+MKN^2)=\mathcal{O}(M^2KN^4)$. If parallel processing across UEs is available, the effective wall-clock per-iteration complexity remains $\mathcal{O}(M^2N^4)$.

For the network-wide schemes, i.e., Algorithm~\ref{algo_2} and Algorithm~\ref{algo_4}, by exploiting the block-diagonal structure of $\bm{\Gamma}$, the dominant cost is similarly the gradient computation, which has complexity $\mathcal{O}(M^2KN^4)$, while the proximal update and the FISTA extrapolation have $\mathcal{O}(MKN^2)$ computational complexity. Therefore, the per-iteration complexity of the network-wide schemes is $\mathcal{O}(M^2KN^4+MKN^2)=\mathcal{O}(M^2KN^4)$, which is again dominated by the gradient step. Overall, under serial implementation, all four sparse activation algorithms have the same per-iteration complexity, namely $\mathcal{O}(M^2KN^4)$, since the gradient computation is the dominant step. Their practical computational burden, however, is still different: the UE-specific schemes are more implementation-friendly due to parallelizability across UEs, while the antenna-level schemes involve finer group structures than the array-level schemes.

\vspace{-0.4cm}
\subsection{Convergence Analysis}
In this part, we deliver the convergence proof for the proposed sparse activation algorithms. We notice that sparse optimization problems solved by Algorithm~\ref{algo_1} to Algorithm~\ref{algo_4} all belong to the standard composite convex form $\underset{ \underline{\mathbf{w}}}{\min}\,\,F( \underline{\mathbf{w}})\triangleq f( \underline{\mathbf{w}})+\Omega ( \underline{\mathbf{w}})$, which is exactly the standard formulation considered in \cite[Sec.~5.3.3.1]{rish2014sparse}. Here, the smooth term is the quadratic function
$f(\underline{\mathbf{w}})=\underline{\mathbf{w}}^T\underline{\mathbf{\Gamma}}\underline{\mathbf{w}}-2\underline{\mathbf{w}}^T\underline{\bm{\xi}}$
while the nonsmooth term $\Omega(\underline{\mathbf{w}})$ is either a standard group $\ell_{1}/\ell_{2}$ penalty for Algorithm~\ref{algo_1} and Algorithm~\ref{algo_3}, or a tree-structured hierarchical $\ell_{1}/\ell_{2}$ penalty for Algorithm~\ref{algo_2} and Algorithm~\ref{algo_4}. Since $\underline{\mathbf{\Gamma}}\succeq \mathbf 0$, the quadratic function $f(\underline{\mathbf{w}})$ is convex, and its gradient $\nabla f(\underline{\mathbf{w}})=2\underline{\mathbf{\Gamma}}\underline{\mathbf{w}}-2\underline{\bm{\xi}}$ is affine and therefore Lipschitz continuous. Hence, the smooth part of all four problems satisfies the basic assumption required by the proximal-gradient and FISTA framework, as shown in \cite[Sec.~5.3.3]{rish2014sparse}.

Moreover, for Algorithm~\ref{algo_1} and Algorithm~\ref{algo_3}, the update is based on a standard group $\ell_{1}/\ell_{2}$ regularizer, whose proximal operator is available in closed-form through group soft-thresholding applied independently to each group, as in \cite[Sec.~5.3.3.3]{rish2014sparse}. Therefore, the group-wise updates in Algorithm~\ref{algo_1} and Algorithm~\ref{algo_3} compute the exact proximal operator of the corresponding nonsmooth penalty. For Algorithm~\ref{algo_2} and Algorithm~\ref{algo_4}, the involved groups are tree-structured, namely, any two groups are either disjoint or one is included in the other. In this case, the overall proximal operator is again exact and can be written as an ordered composition of group-wise proximal operators, as shown in \cite[Eq.~(3.8)]{bach2012optimization}. Therefore, the leaf-to-root updates in Algorithm~\ref{algo_2} and Algorithm~\ref{algo_4} also amount to an exact proximal step.
Consequently, Algorithms~\ref {algo_1} to ~\ref{algo_4} are all instances of the standard proximal-gradient method for convex composite optimization, and their accelerated versions are standard FISTA-type implementations. By the classical convergence results of proximal methods in \cite{rish2014sparse} and \cite{bach2012optimization}, the non-accelerated proximal-gradient scheme enjoys the global objective convergence rate $\mathcal{O}(1/i)$, whereas the accelerated FISTA scheme improves it to $\mathcal{O}(1/i^{2})$, as provided in \cite[Sec.~5.3.3.2]{rish2014sparse} and \cite[Sec.~3.2]{bach2012optimization}, with $i$ being the iteration number. It is worth emphasizing that the $\mathcal{O}(1/i^{2})$ result is the convergence rate of the objective values, which is the standard guarantee for accelerated proximal methods.

\vspace{-0.5cm}

\subsection{Fronthaul Load Saving}
In the considered processing architecture, the fronthaul transmission is mainly associated with the local soft estimates delivered from the serving APs to the CPU. To explicitly quantify the fronthaul load reduction brought by the proposed sparse activation schemes, we denote the fronthaul load per coherence block as
$\mathcal{B} _{\mathrm{fh}}=2\zeta _{\mathrm{fh}}\tau _u\sum_{m=1}^M{\sum_{k=1}^K{b_{mk}}}$, where $\zeta_{\mathrm{fh}}$ denotes the number of quantization bits for each real or imaginary part, $\tau_u$ is the number of uplink data symbols within one coherence block, and $b_{mk}\in\{0,1\}$ indicates whether AP $m$ serves UE $k$, as defined in Sec.~\ref{sec_power_model}. This expression shows that the fronthaul load can be reduced by deactivating AP-UE serving links through sparse activation, which is reflected by $b_{mk}=0$. By contrast, under the full-activation benchmark, all APs serve all UEs, and the corresponding fronthaul load becomes $\mathcal{B}_{\mathrm{fh}}^{\mathrm{full}}=2\zeta_{\mathrm{fh}}\tau_u MK$. Accordingly, the fronthaul bits saved by the sparse activation scheme can be expressed as $\Delta \mathcal{B}_{\mathrm{fh}}=\mathcal{B}_{\mathrm{fh}}^{\mathrm{full}}-\mathcal{B}_{\mathrm{fh}}=2\zeta_{\mathrm{fh}}\tau_u(MK-\sum_{m=1}^{M}\sum_{k=1}^{K} b_{mk})$.

% \begin{figure*}[htbp]
% 	\centering
% 	\begin{minipage}{0.32\linewidth}
% 		\centering
%         \includegraphics[width=1.05\columnwidth]{FIG_OBE_N_taup5.eps}\vspace*{-0.3cm}
%         \caption{Average SE per UE for the OBE weighting and LSFD weighting schemes over different $N$ with $M=10$ and $K=10$. \label{Fig_OBE_N}
%         \vspace*{+0.05cm}}
% 	\end{minipage}
% 	%\qquad
% 	\begin{minipage}{0.32\linewidth}
% 		\centering
%         \includegraphics[width=1.05\columnwidth]{FIG_OBE_taup.eps}\vspace*{-0.3cm}
%         \caption{Average SE per UE for the OBE weighting and LSFD weighting schemes over different $\tau_p$ with $M=10$, $K=10$, and $N=9$. \label{Fig_OBE_taup}
%         \vspace*{+0.05cm}}
% 	\end{minipage}
%  \begin{minipage}{0.32\linewidth}
% 		\centering
%         \includegraphics[width=1.05\columnwidth]{FIG_Antenna_Per_UE_SE_EE_251222.eps}\vspace*{-0.3cm}
%         \caption{Average SE per UE and average system EE for the antenna-level activation for the particular UE in Sec.~\ref{sec_antenna_per_UE} with $M=10$, $K=10$, and $N=9$, where the dots are drawn with the relative sparsity ratio $\bar{\lambda}$ varying from $0.05$ to $1$ from right to left with the step $0.05$.
%         \label{Fig_Antenna_per_UE}}
% 	\end{minipage}\vspace*{-0.3cm}
% \end{figure*}

\vspace{-0.2cm}
\section{Numerical Results}
It is worth noting that the technical derivations and the proposed optimization framework are not limited to the specific AP-UE deployment strategy and channel model considered in the simulations. They are applicable to more general AP-UE geometries and channel models as long as the required channel statistics are available. Therefore, the following results are provided only as a representative numerical example.
We will evaluate the potential power savings from antenna and array inactivation by considering a CF mMIMO network with the coverage area being $200\times 200 \, \mathrm{m}^2$, where the center of the coverage area is defined as the origin and the wall, in which the APs are uniformly deployed, is located in the $x-y$ plane with the coordinate in the $z$-axis being $-100\, \mathrm{m}$. Note that the $x$-axis and $y$-axis are parallel and perpendicular to the coverage area, respectively, and the $y$-axis coordinate for the coverage area is $0$. UEs are randomly distributed in the coverage area. We define the coordinates of the bottom left point of the $m$-th array as $\mathbf{r}_m=[x_m,y_m,z_m]^T\in \mathbb{R} ^3$ with $y_m=L_{AP}$  being the array height. Counting row-by-row from the bottom left, we can derive the coordinates of the $n$-th antenna of the $m$-th array as $\mathbf{r}_{mn}=[r_{mn,x},r_{mn,y},r_{mn,z}]^T=[x_m+\mathrm{mod(}n-1,N_H)\Delta _H,L_{BS}+\lfloor (n-1)/N_H \rfloor \Delta _H,z_m]^T\in \mathbb{R} ^3$, where $\Delta _H$ and $\Delta _V$ represent the horizontal and vertical antenna spacing, respectively. Moreover, we denote the coordinates of the $k$-th UE as $\mathbf{s}_k=[s_{k,x},s_{k,y},s_{k,z}] ^T\in \mathbb{R} ^3$ with $s_{k,y}=L_{UE}$ being the UE height.
We consider the widely studied half-wavelength antenna spacing scenario, where $\Delta _H=\Delta _V=\bar{\lambda}/2$ with $\bar{\lambda}$ being the wavelength. For the channel modeling, we  model $\overline{\mathbf{h}}_{mk}$ based on the standard spherical wave channel model \cite{2023arXiv231011044L,CM25CKJ} as 
$
\overline{\mathbf{h}}_{mk}=\sqrt{\beta _{mk}^{\mathrm{LoS}}}[ e^{-j\frac{2\pi}{\bar{\lambda}}(d_{m1k}-d_{mk})},\dots ,e^{-j\frac{2\pi}{\bar{\lambda}}(d_{mNk}-d_{mk})} ] ^T,
$
where $d_{mnk}=\left\| \mathbf{r}_{mn}-\mathbf{s}_k \right\|$  represents the distance between the $n$-th antenna of the $m$-th array and the $k$-th UE, $d_{mk}=\left\| \mathbf{r}_{m}-\mathbf{s}_k \right\|$ denotes the reference distance between the $m$-th array and the $k$-th UE, and $\beta _{mk}^{\mathrm{LoS}}$ is the large-scale fading coefficient between the $m$-th array and the $k$-th UE. As for $\mathbf{R}_{mk}$, we assume that an isotropic scattering environment is considered and the $(n_1,n_2)$-th element of $\mathbf{R}_{mk} $ can be modeled based on \cite{9300189} as 
$
\left[ \mathbf{R}_{mk} \right] _{n_1n_2}=\beta _{mk}^{\mathrm{NLoS}}\mathrm{sinc}\left( 2\left\| \mathbf{r}_{mn_1}-\mathbf{r}_{mn_2} \right\| /\bar{\lambda} \right)
$
where $\mathbf{r}_{mn_1}$ denotes the coordinates of the $n_1$-th antenna of the $m$-th array and $\beta _{mk}^{\mathrm{NLoS}}$ denotes the large-scale fading coefficient for the NLoS component. Moreover, we have $\beta _{mk}^{\mathrm{LoS}}=\frac{\kappa _{mk}}{\kappa _{mk}+1}\beta _{mk}$ and $\beta _{mk}^{\mathrm{NLoS}}=\frac{1}{\kappa _{mk}+1}\beta _{mk}$, where $\kappa _{mk}=10^{1.3-0.003d_{mk}}$ denotes the Rician $\kappa$ factor between AP $m$ and UE $k$ and $\beta _{mk}$ is the large-scale fading coefficient  between AP $m$ and UE $k$, which can be modeled based on the 3GPP COST 231 Walfish-Ikegami model as \cite[Eq. (5.2-3)]{3GPP2024}. Besides, we have $L_{AP}=12.5 \, \mathrm{m}$ and $L_{UE}=1.5 \, \mathrm{m}$. We have the carrier frequency $f_c=1.9 \, \mathrm{GHz}$, the bandwidth $\mathcal{B}=20\, \mathrm{MHz}$, $\tau _c=200
$, $\tau _p=5$ (otherwise mentioning), $\sigma ^2=-94 \, \mathrm{dBm}$. For the parameters in the power consumption model, unless mentioned, for the UE power consumption model, we have $P_{k}^{\mathrm{ue},\mathrm{c}}=0.1\, \mathrm{W}$,  $p_{k,p}=p_k=0.2 \,\mathrm{W}$, and $\eta _{\mathrm{ue}}=0.6$. For the AP power consumption, we have $P_{m}^{\mathrm{array}}=4 \, \mathrm{W}$, $P_{mn}^{\mathrm{ant},\mathrm{c}}=0.2\, \mathrm{W}$, and $P_{mnk}^{\mathrm{pro}}=0.4 \,\mathrm{W}$. For the consumed power in the fronthaul links, we have $P_{m}^{\mathrm{fh},\mathrm{fix}}=0.4 \, \mathrm{W}$ and $P_{mk}^{\mathrm{sig}}=0.01\, \mathrm{W}$. For the power consumption in the CPU, we have $P^{\mathrm{cpu},\mathrm{fix}}=2 \,\mathrm{W}$ and $P^{\mathrm{cpu},\mathrm{dec}}=0.8 \, \mathrm{W}/\left( \mathrm{Gbit}/\mathrm{s} \right) $ motivated from the parameter choice in \cite{8187178,10113887}.

As for the choice of simulated sparse weighting schemes, we exploit the fact that the sparse optimization algorithms described above generate weighting vectors in which certain elements are driven to zero. These sparsity patterns directly indicate the feasible antenna/array-UE activation modes. Consequently, we can tailor the optimal weighting scheme in \eqref{opt_weighting} by setting the entries corresponding to inactive antenna/array-UE modes to zero, and employ these tailored weighting schemes to evaluate the achievable performance. Thus, the sparse optimization algorithms mainly serve as a mechanism for exploring the activation modes. Moreover, both the L-MMSE and L-MR combining schemes are considered as representative local processing strategies. The L-MMSE combining serves as a stronger benchmark because of its superior local interference suppression capability, whereas the L-MR combining is a lower-complexity scheme that is sensitive to strong multi-user interference. Nevertheless, retaining the L-MR combining is still interesting, since it helps reveal the behavior of the proposed sparse activation schemes under a weaker local combining architecture and thus complements the results obtained under the stronger L-MMSE benchmark.
In the following, we first validate the performance gain of the OBE weighting scheme, then investigate the sparse activation behaviors at the antenna and array levels separately, and finally provide an overall comparison together with the studies for the impact of different system and modeling settings. Note that the following simulation results are obtained via Monte-Carlo methods over $30$ random UE-location setups, each with $800$ random channel realizations.

\vspace{-0.3cm}
\subsection{OBE Weighting Gain}

Fig.~\ref{Fig_OBE_all} (a) shows the average SE for the OBE weighting and LSFD weighting schemes for different $N$. We observe that the OBE scheme consistently outperforms the LSFD scheme, and the SE improvement for the OBE scheme compared to the LSFD scheme becomes larger as $N$ increases. This is because the OBE scheme can efficiently exploit the spatial correlation to design the antenna-specific weighting scheme to enhance the performance, while the LSFD scheme only utilizes the scalar weighting coefficient for all antennas in each array. Moreover, the OBE scheme yields a larger SE gain over LSFD when L-MR combining is used than when L-MMSE combining is used. This is because the L-MR combining has a weaker inherent ability to suppress interference. The involvement of the OBE scheme can help the L-MR combining to conquer this deficiency to also achieve excellent SE performance. We also observe that the combination of the L-MR combining and OBE scheme can achieve performance approaching the combination of the L-MMSE combining and LSFD scheme, highlighting that the powerful system-level cooperation enabled by the OBE can make the low-complexity local processing scheme, such as the L-MR, competitive with substantially more complex baselines, such as the L-MMSE.

Fig.~\ref{Fig_OBE_all} (b) investigates the impact of $\tau_p$ on the SE performance. As $\tau_p$ increases, the pilot contamination degrades, and thus, the quality of channel estimates becomes higher. For the L-MMSE combining, we can observe that the performance gap between the LSFD and OBE weighting schemes becomes smaller as $\tau_p$ increases. This is because the L-MMSE combining has a more powerful interference suppression capability. Thus, the OBE scheme performs more effectively in the scenario with severe pilot contamination. For the L-MR combining, the performance gap between the LSFD and OBE schemes becomes larger as $\tau_p$ increases, which demonstrates that the OBE scheme can more efficiently utilize the benefit of high-quality channel estimates than the LSFD scheme under the local processing scheme with poor interference suppression ability. However, the OBE scheme cannot ``save" the L-MR combining in the scenario with high-quality channel estimates, where the performance gap between the L-MR and L-MMSE combining increases as $\tau_p$ increases. This indicates that once CSI is sufficiently accurate, the performance ceiling is primarily decided by the interference suppression capability of the local processing scheme.

\begin{figure}[t]
\centering
\includegraphics[width=0.8\columnwidth]{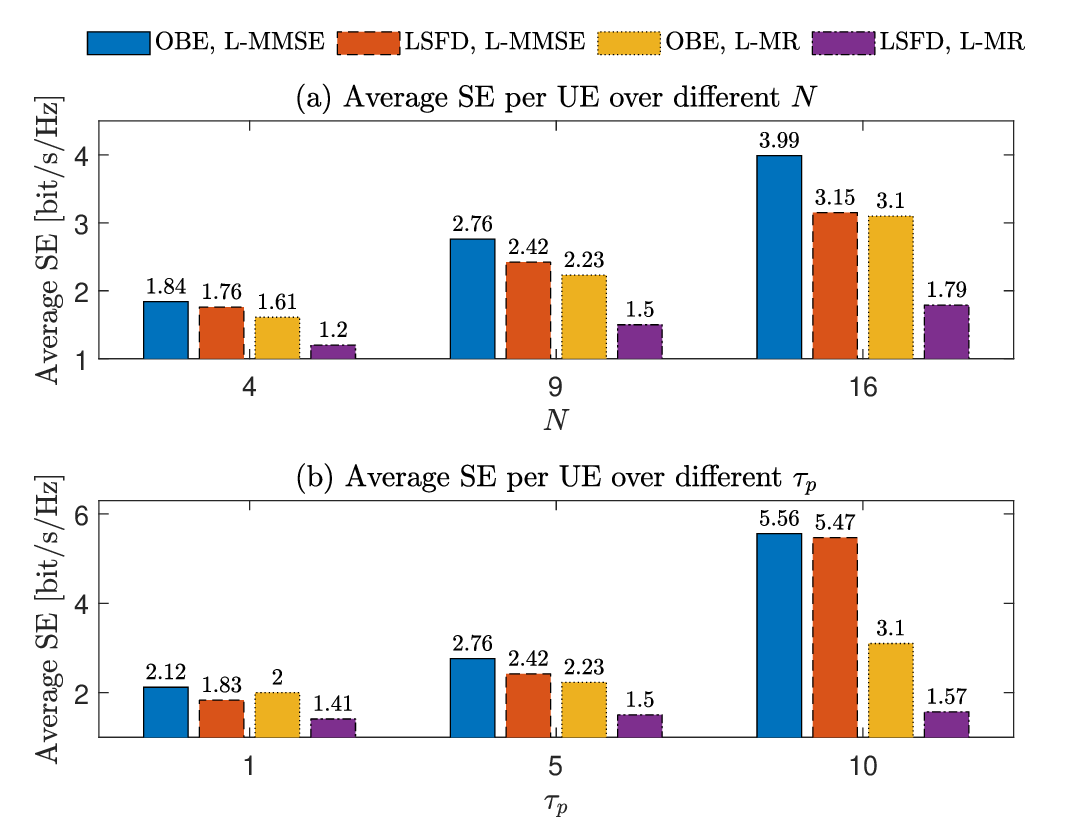}\vspace*{-0.3cm}
\caption{Average SE per UE for the OBE and LSFD weighting schemes over different $N$ with $M=10$ and $K=10$ in subplot (a) and over different $\tau_p$ with $M=10$, $K=10$, and $N=9$ in subplot (b), respectively. \label{Fig_OBE_all}
\vspace{-0.15cm}}
\end{figure}

\begin{figure}[t]
\centering
\includegraphics[width=0.8\columnwidth]{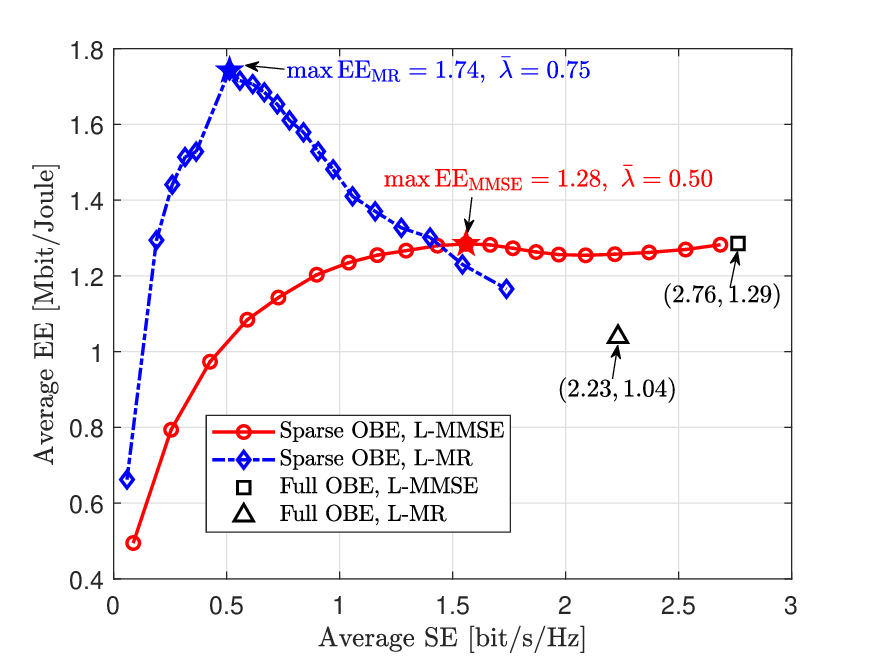}\vspace*{-0.3cm}
\caption{Average SE per UE and average system EE for the antenna-level activation for the particular UE in Sec.~\ref{sec_antenna_per_UE} with $M=10$, $K=10$, and $N=9$, where the dots are drawn with the relative sparsity ratio $\bar{\lambda}$ varying from $0.05$ to $1$ from right to left with the step $0.05$.
\label{Fig_Antenna_per_UE}}
\vspace{-0.25cm}
\end{figure}

\begin{figure}[t]\centering
\subfigure[L-MMSE combining]{
\begin{minipage}{8cm}\centering
\includegraphics[width=0.8\columnwidth]{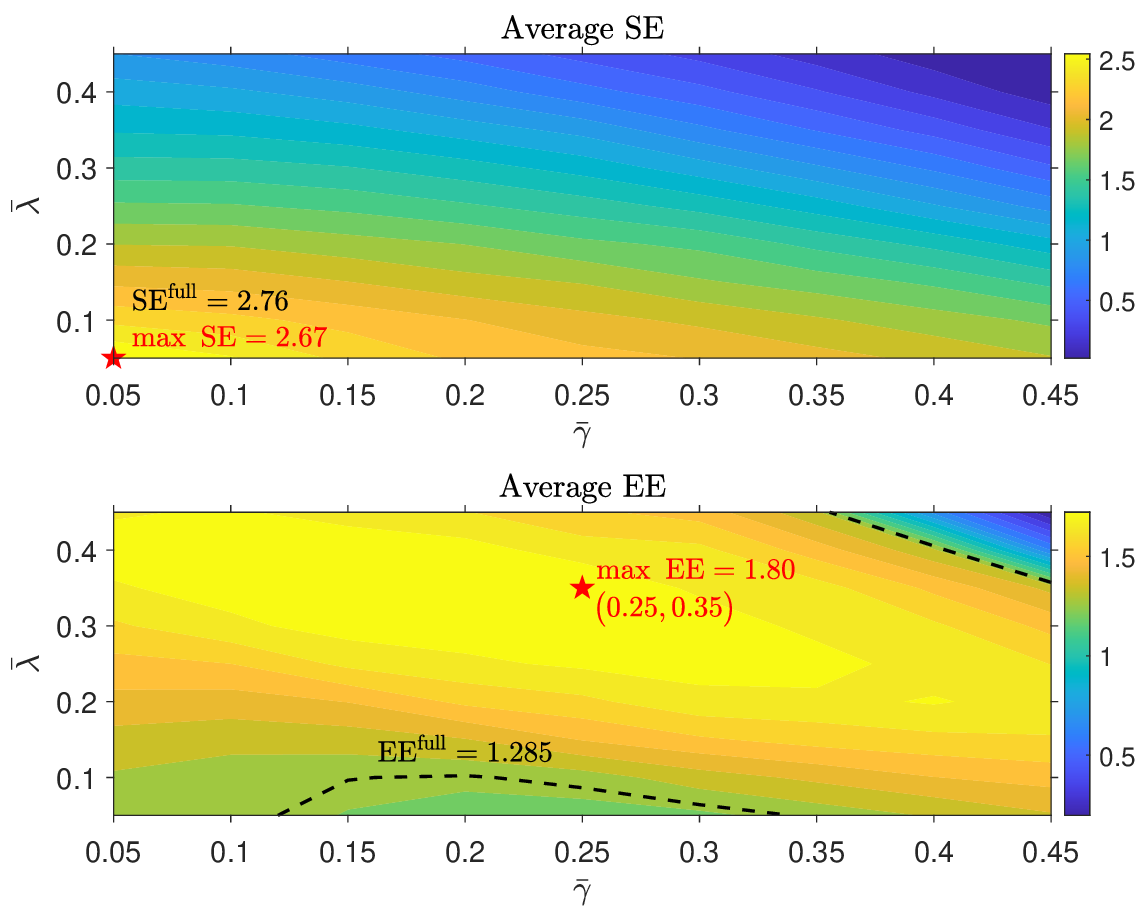}
\end{minipage}}
\subfigure[L-MR combining]{
\begin{minipage}{8cm}\centering
\includegraphics[width=0.8\columnwidth]{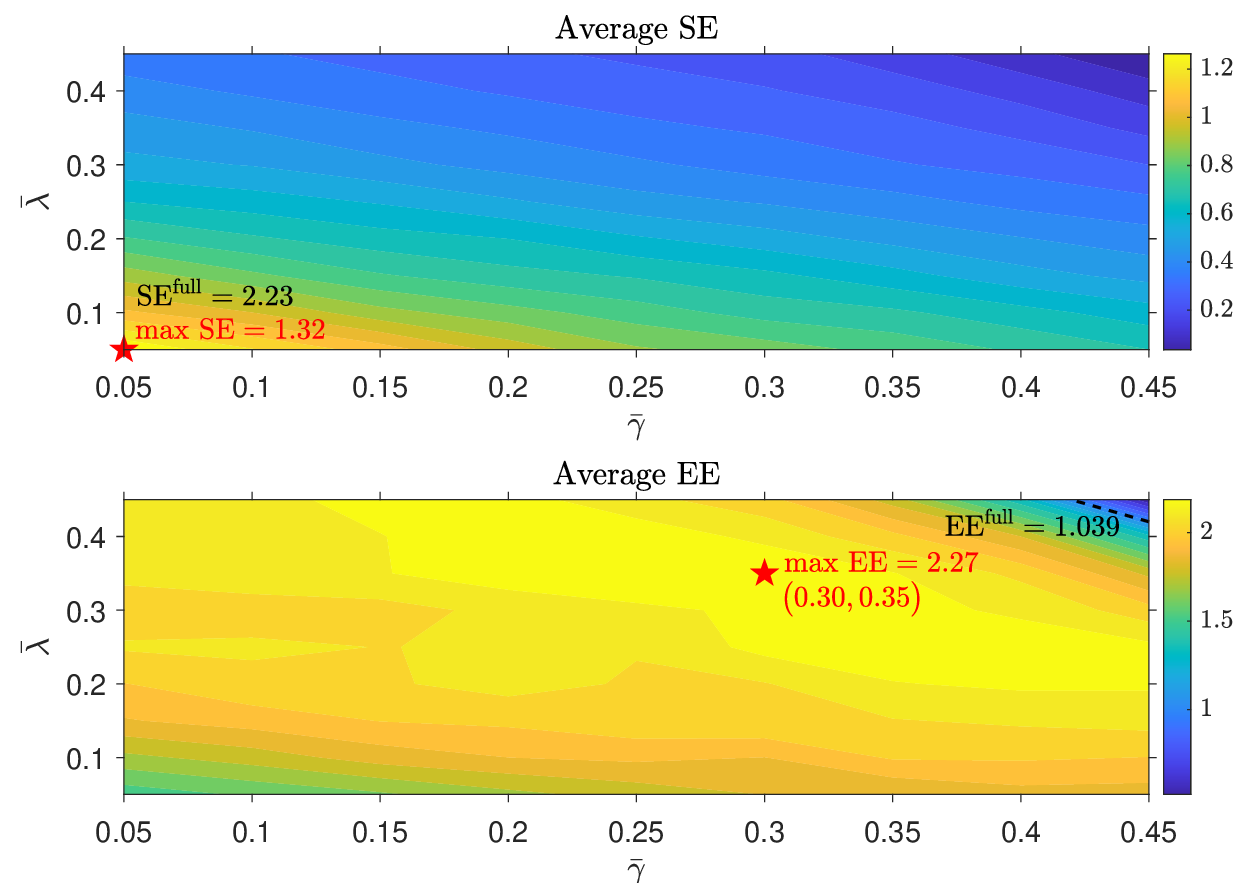}
\end{minipage}}\vspace{-0.25cm}
\caption{Average SE per UE and average system EE for the antenna-level activation for all UEs in Sec.~\ref{antenna_all_UE} with different values of the relative sparsity ratios $\{\bar{\gamma},\bar{\lambda}\}$ varying from $0.05$ to $0.45$ with the step $0.05$ over the L-MMSE and L-MR combining with $M=10$, $K=10$, and $N=9$. \label{Fig_antenna_all}\vspace*{-0.2cm}}
\end{figure}

\begin{figure}[t]
\centering
\includegraphics[width=0.82\columnwidth]{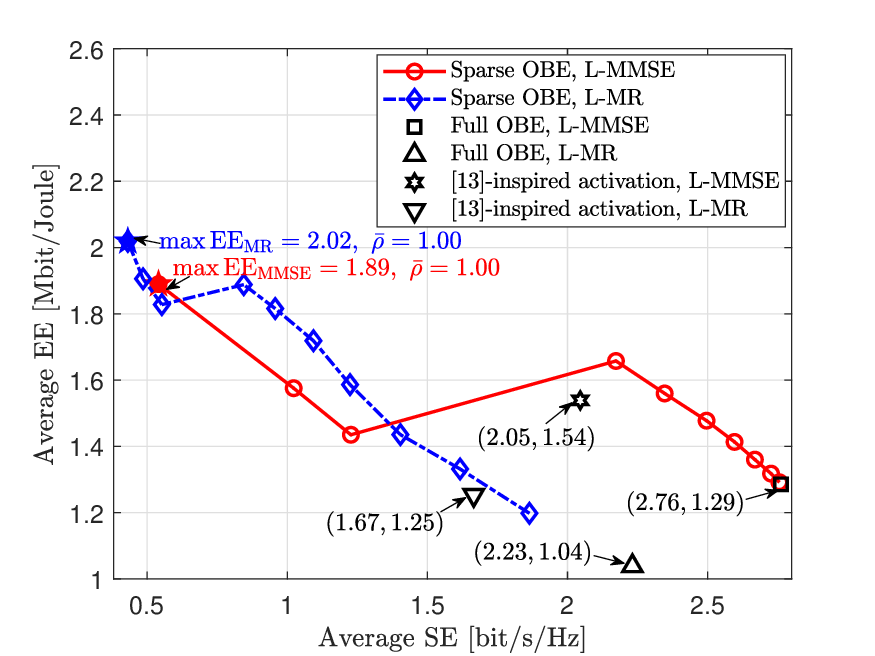}\vspace{-0.4cm}
\caption{Average SE per UE and average system EE for the array-level activation for the particular UE in Sec.~\ref{sec_antenna_per_UE} with $M=10$, $K=10$, and $N=9$, where the dots are drawn with the relative sparsity ratio $\bar{\rho}$ varying from $0.1$ to $1$ from right to left with the step $0.1$. \label{Fig_array_per}}
\vspace{-0.3cm}
\end{figure}

\vspace{-0.5cm}
\subsection{Antenna-Level Sparse Activation Evaluation}
\vspace{-0.2cm}
Fig.~\ref{Fig_Antenna_per_UE} studies the SE-EE performance for the sparse antenna-level activation scheme for the particular UE, introduced in Sec.~\ref{sec_antenna_per_UE}. We can observe that, as $\bar{\lambda}$ decreases, which can be reflected from the dots from left to right, the EE first increases and then decreases or stays approximately stationary for L-MMSE or L-MR combining, respectively. We can observe that with L-MMSE combining, when $\bar{\lambda}$ is small, the EE for the sparse activation scheme remains approaching that of the full activation scheme. This is because although some antenna-UE links are inactivated, the decrease of power consumption induced by this sparsity is not dominant in the total power consumption, so that the EE may not be clearly benefited due to the consequent SE decrease. L-MR is more sparsity-sensitive than L-MMSE because it relies mainly on coherent array gain and offers little interference suppression, and thus, deactivating even a small fraction of antennas significantly reduces the desired signal gain, yielding a sharp SE degradation. But the L-MR combining can achieve significant EE improvement over the benchmark when the sparsity level is properly selected, where many antenna-UE links, even the whole antenna or array, can be turned off, and thus, the power consumption significantly degrades.

% \begin{figure}[t]
% \centering
% \includegraphics[width=0.825\columnwidth]{FIG_Array_Per_UE_SE_EE_251222.eps}\vspace{-0.4cm}
% \caption{Average SE per UE and average system EE for the array-level activation for the particular UE in Sec.~\ref{sec_antenna_per_UE} with $M=10$, $K=10$, and $N=9$, where the dots are drawn with the relative sparsity ratio $\bar{\rho}$ varying from $0.1$ to $1$ from right to left with the step $0.1$. \label{Fig_array_per}}
% \vspace{-0.3cm}
% \end{figure}

\begin{figure}[t]\centering
\vspace{0.3cm}
\subfigure[L-MMSE combining]{
\begin{minipage}{8cm}\centering
\includegraphics[width=0.8\columnwidth]{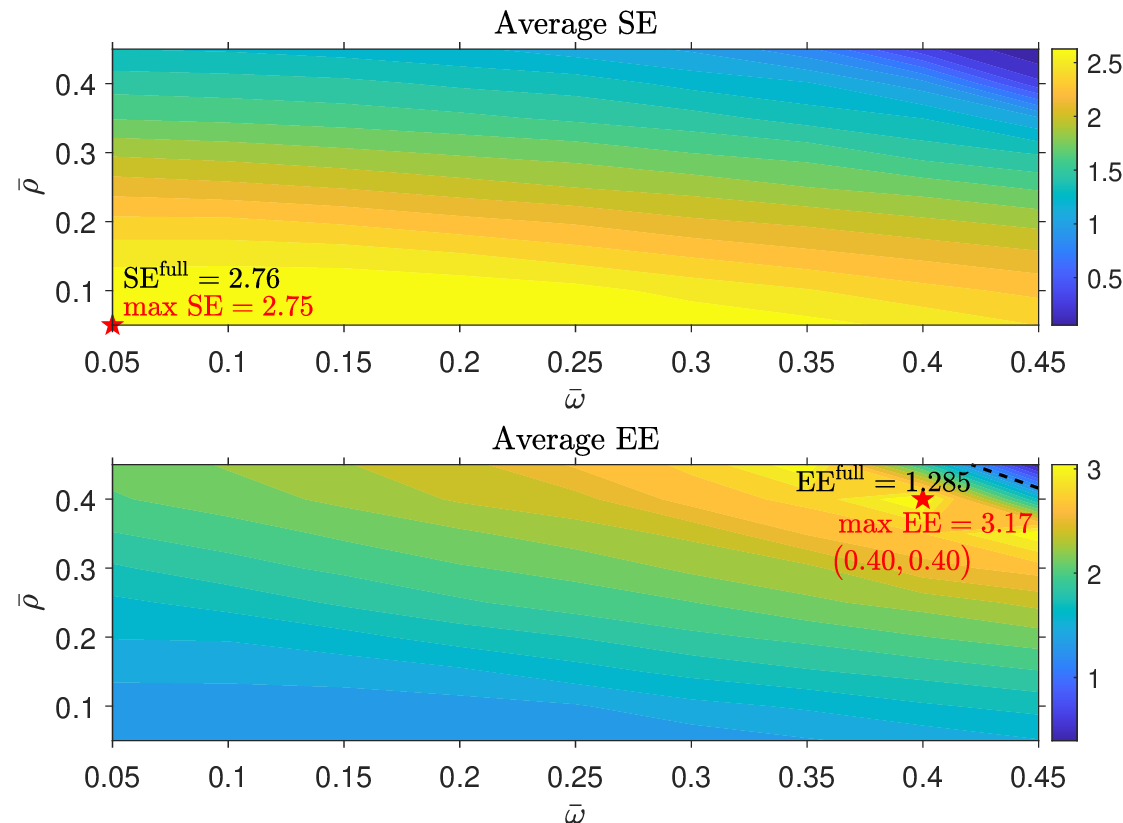}\vspace{+0.1cm}
\end{minipage}}
\subfigure[L-MR combining]{
\begin{minipage}{8cm}\centering
\includegraphics[width=0.8\columnwidth]{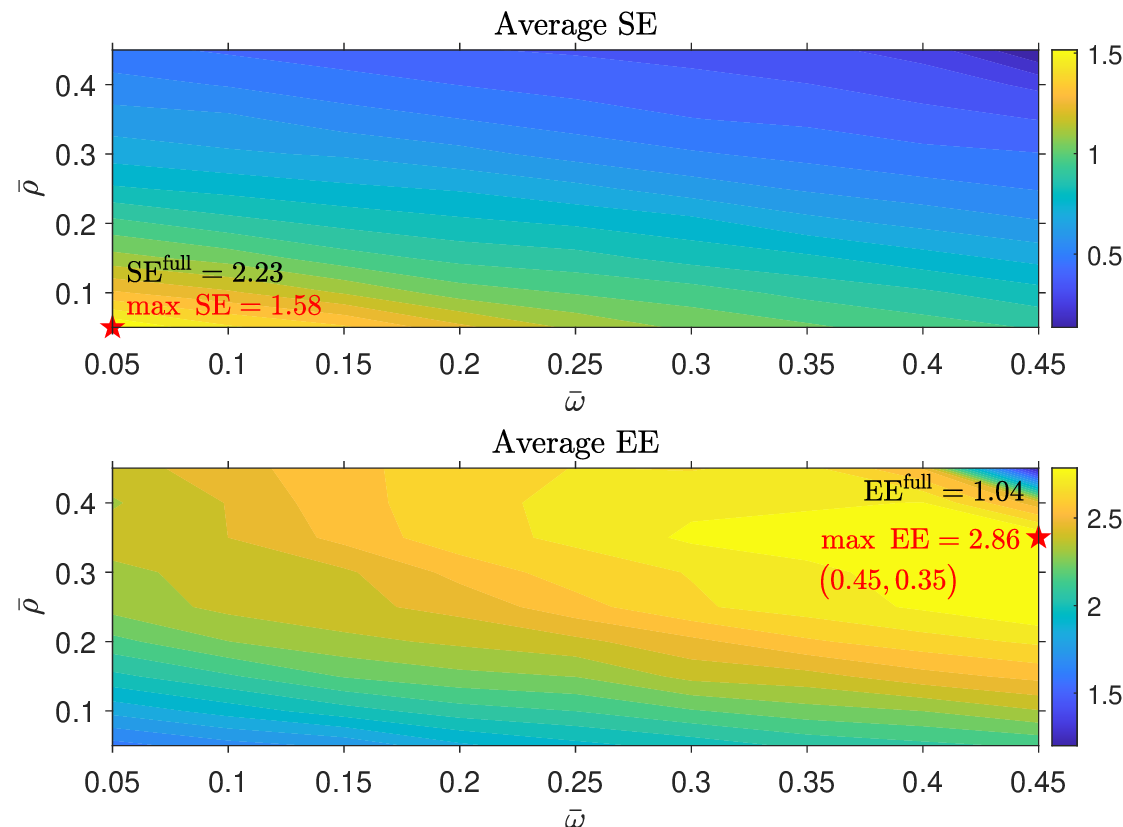}\vspace{+0.1cm}
\end{minipage}}\vspace{-0.25cm}
\caption{Average SE per UE and average system EE for the array-level activation for all UEs in Sec.~\ref{array_all_UE} with different values of the relative sparsity ratios $\{\bar{\omega},\bar{\rho}\}$ varing from $0.05$ to $0.45$ with the step $0.05$ over the L-MMSE and L-MR combining with $M=10$, $K=10$, and $N=9$. \label{Fig_array_all}
\vspace{-0.15cm}}
\end{figure}

Fig.~\ref{Fig_antenna_all} illustrates the average SE per UE and average system EE for the sparse antenna-level activation scheme for all UEs in Sec.~\ref{antenna_all_UE}. For both the L-MMSE and L-MR combining, the average SE decreases monotonically with either $\bar{\gamma}$ or $\bar{\lambda}$, since increasing sparsity inactivates more antennas and antenna–UE links. This degradation is significantly more severe for L-MR, consistent with Fig.~\ref{Fig_Antenna_per_UE}. The average EE exhibits a unimodal behavior and attains its maximum at moderate sparsity levels, where the peak EE substantially exceeds the full-activation benchmark for both combining schemes. This indicates that a moderate increase of $\{\bar{\gamma},\bar{\lambda}\}$ effectively reduces antenna-level power consumption faster than the corresponding SE degradation, whereas overly aggressive sparsification (large $\bar{\gamma}$ and $\bar{\lambda}$) causes a sharp EE drop due to severely degraded SE. Moreover, we can observe that the SE degradation and the EE improvement are predominantly governed by the leaf sparse factor $\bar{\lambda}$, where the SE and EE variations along the $\bar{\lambda}$ axis are stronger than those of the $\bar{\gamma}$ axis. The EE enhancement capability of this scheme is more effective than the scheme in Sec.~\ref{sec_antenna_per_UE} since the power consumption reduction in this scheme is more notable across the total power consumption.

\begin{figure}[t]\centering
\vspace{0.3cm}
\subfigure[Achievable performance]{
\begin{minipage}{8cm}\centering
\includegraphics[width=0.9\columnwidth]{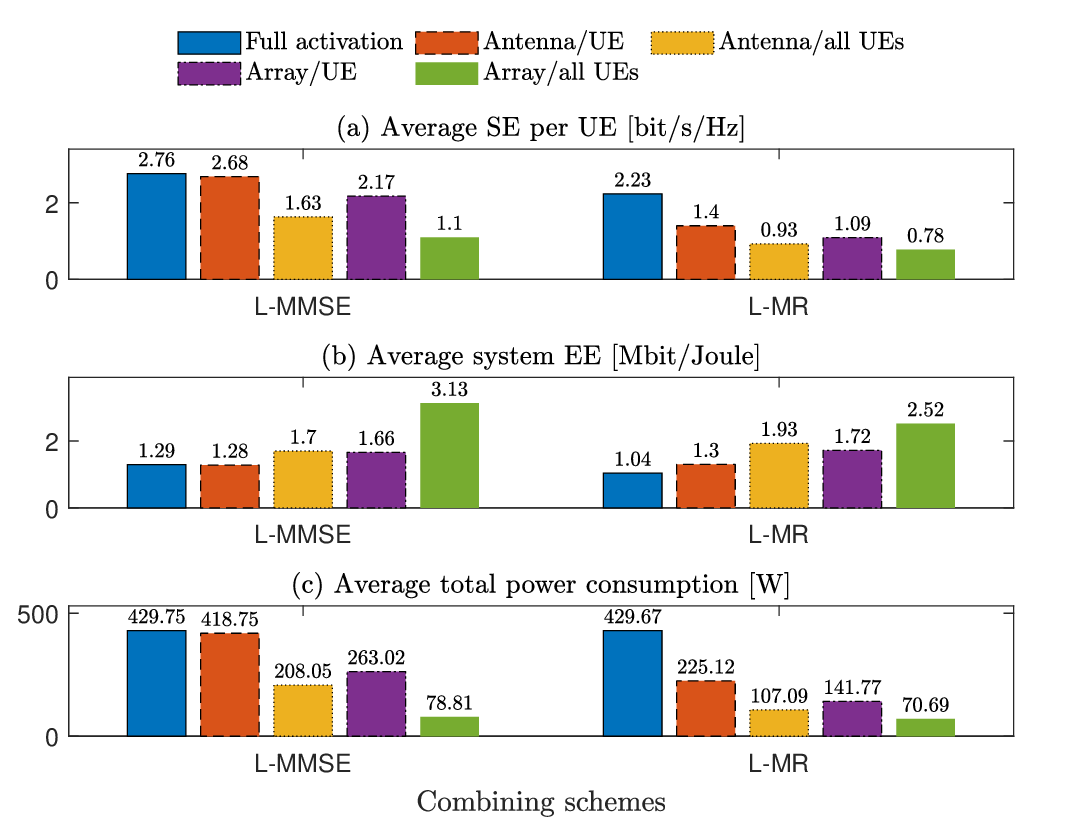}
\end{minipage}}
\subfigure[Antenna/array resources]{
\begin{minipage}{8cm}\centering
\includegraphics[width=0.9\columnwidth]{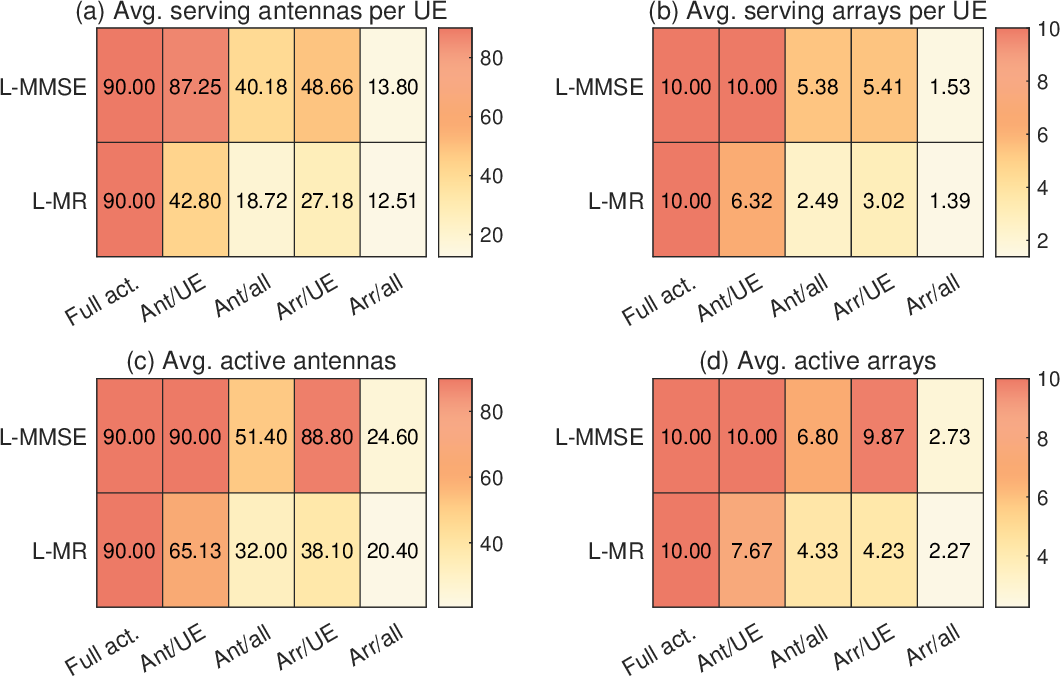}
\end{minipage}}
\caption{Comparison for achievable performance and activated antenna/array resources for all studied activation schemes with respective optimal sparsity ratios maximizing the EST metric over $M=10$, $K=10$, and $N=9$. ``Antenna/UE", ``Array/UE", ``Antenna/all UEs", and ``Array/all UEs" represent the antenna-level activation for particular UE, array-level activation for particular UE, antenna-level activation for all UEs, and array-level activation for all UEs, respectively. The terms in the subplot (b) are abbreviations for those in the subplot (a). \label{Fig_compare_all} 
\vspace{-0.15cm}}
\end{figure}

\vspace{-0.5cm}

\subsection{Array-Level Sparse Activation Evaluation}

Fig.~\ref{Fig_array_per} studies the average SE-EE trade-off per UE for the sparse array-level activation scheme for the particular UE in Sec.~\ref{array_per_UE}. We observe that, under the adopted power consumption model, the EE performance for this scheme is power-dominated for this array-level activation: increasing the sparsity ratio $\bar{\rho}$ reduces the total consumed power faster than the induced SE degradation so the energy efficiency keeps improving on the whole and the maximizer is pushed to the most aggressive sparsity, that is $\bar{\rho}=1$. Meanwhile, this array-level activation scheme achieves a more favorable SE-EE balance than the sparse antenna-level activation for the particular UE in Sec.~\ref{sec_antenna_per_UE}, especially under the L-MMSE combining. Along the SE axis, most L-MMSE points remain close to the full-activation benchmark, while the EE is significantly improved. For instance, $\bar{\rho}=0.5$ can achieve $21\%$ EE improvement compared to the benchmark with only $14\%$ SE loss. This favorable trade-off is because array-level sparsification with properly chosen $\bar{\rho}$ retains the most efficient array-UE links and thus preserves powerful interference suppression capability within the array, while the inactivated array-UE links substantially reduce the array- and fronthaul-related power consumption. For further comparison, we also include the \cite{9064545}-inspired AP array-UE pairing baseline, where the pilot-aware dynamic cooperation clustering rule in~\cite{9064545} is used to generate a scalable array-UE serving pattern while inactivating some AP-UE serving links. As observed, this baseline also improves the EE over the full-activation benchmark, which confirms that removing inefficient array-UE links is beneficial for reducing the power consumption and enhancing EE. However, its operating point still lies inside the trade-off region achieved by the proposed sparse array-level activation, indicating that the sparse optimization-based design can exploit the SE-EE trade-off more effectively than the pairing strategy in \cite{9064545}.

\vspace{-0.4cm}
\subsection{Overall Comparison and Tradeoff Interpretation}

% \begin{figure}[t]\centering
% \vspace{0.3cm}
% \subfigure[L-MMSE combining]{
% \begin{minipage}{8cm}\centering
% \includegraphics[width=0.82\columnwidth]{FIG_Power_Save_LMMSE.eps}
% \end{minipage}}\vspace{-0.15cm}
% \subfigure[L-MR combining]{
% \begin{minipage}{8cm}\centering
% \includegraphics[width=0.82\columnwidth]{FIG_Power_Save_LMR.eps}
% \end{minipage}}\vspace{-0.3cm}
% \caption{SE degradation between the sparse activation schemes and the full-activation scheme $\Delta\mathrm{SE}=\mathrm{SE}_{\mathrm{full}}-\mathrm{SE}_{\mathrm{sparse}}$ versus the power saving for the sparse activation schemes compared to the full-activation scheme $P_{\mathrm{save}}=P_{\mathrm{full}}-P_{\mathrm{sparse}}$ over $M=10$, $K=10$, and $N=9$. The curves are obtained by remapping the simulation results in Figs.~\ref{Fig_Antenna_per_UE}--\ref{Fig_array_all} into the $(P_{\mathrm{save}},\Delta\mathrm{SE})$ plane. For the UE-specific schemes, the original sweep points are used, whereas for the all-UEs schemes, the plotted curves are drawn from the corresponding Pareto boundary. For visual clarity, very closely spaced Pareto points are omitted and markers are shown only for a subset of the plotted points.  \label{Fig_power_save}
% \vspace{-0.15cm}}
% \end{figure}

\begin{figure}[t]
\centering
\includegraphics[width=0.8\columnwidth]{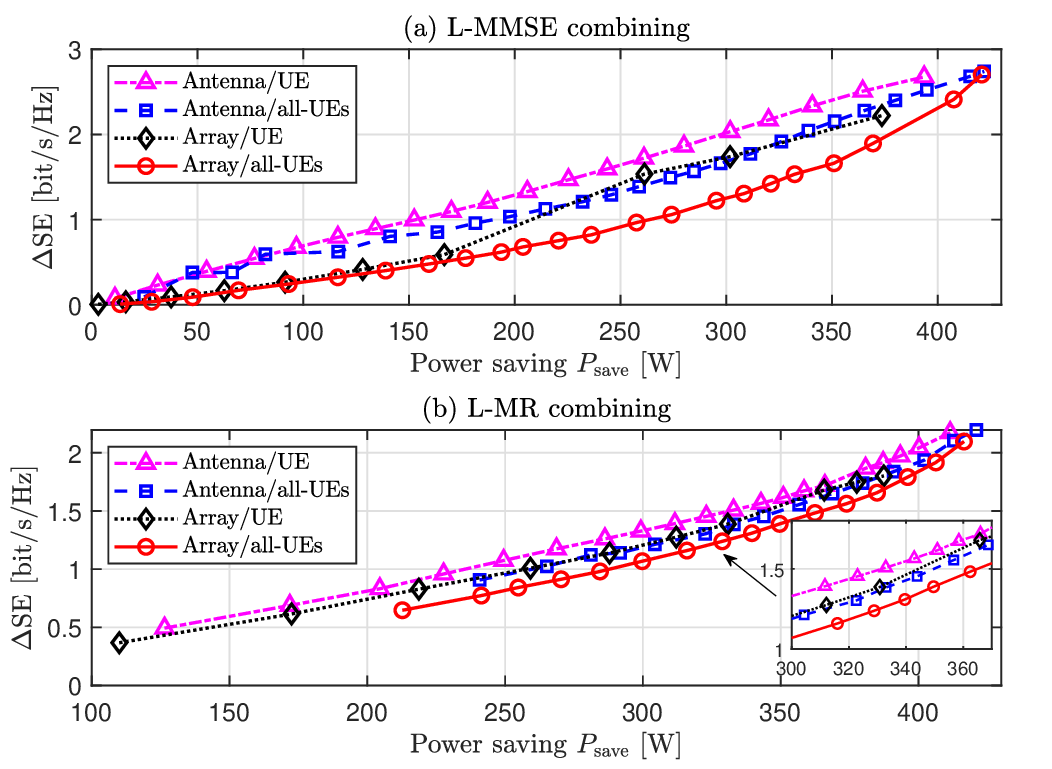}\vspace{-0.4cm}
\caption{SE degradation between the sparse activation schemes and the full-activation scheme $\Delta\mathrm{SE}=\mathrm{SE}_{\mathrm{full}}-\mathrm{SE}_{\mathrm{sparse}}$ versus the power saving for the sparse activation schemes compared to the full-activation scheme $P_{\mathrm{save}}=P_{\mathrm{full}}-P_{\mathrm{sparse}}$ over $M=10$, $K=10$, and $N=9$. The curves are obtained by remapping the simulation results in Figs.~\ref{Fig_Antenna_per_UE}--\ref{Fig_array_all} into the $(P_{\mathrm{save}},\Delta\mathrm{SE})$ plane. For the UE-specific schemes, the original sweep points are used, whereas for the all-UEs schemes, the plotted curves are drawn from the corresponding Pareto boundary. For visual clarity, very closely spaced Pareto points are omitted and markers are shown only for a subset of the plotted points.  \label{Fig_power_save}
\vspace{-0.1cm}}
\end{figure}

\begin{figure}[t]
\centering
\includegraphics[width=0.8\columnwidth]{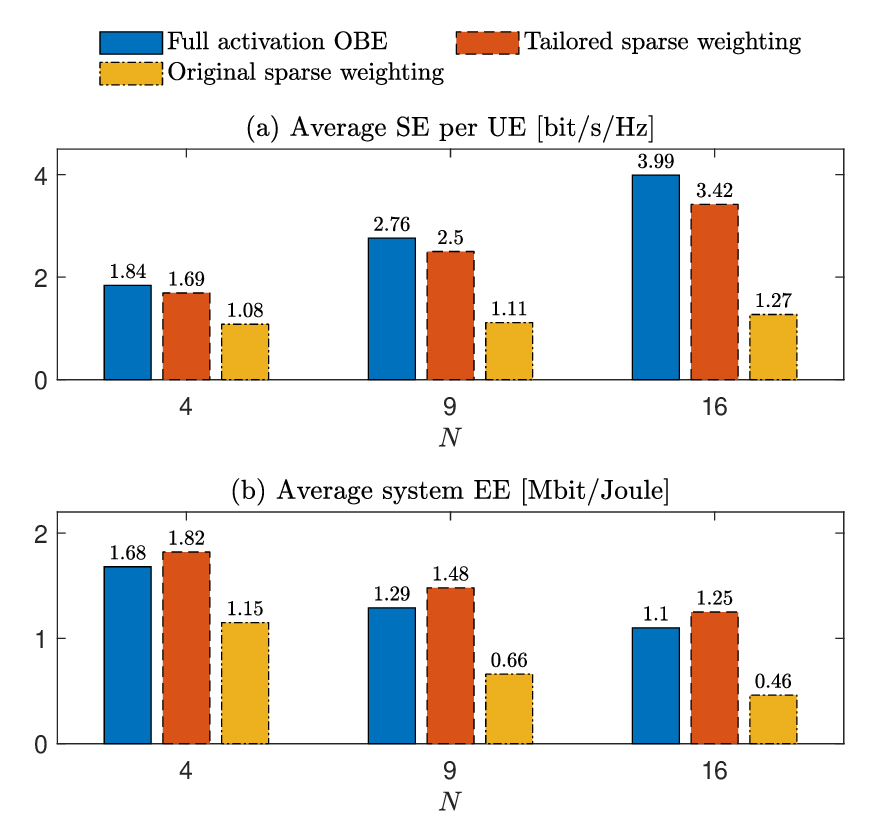}\vspace{-0.4cm}
\caption{Average SE  and average system EE for the array/UE activation scheme over different weighting schemes and different $N$ with $M=10$, $K=10$, and $\bar{\rho}=0.5$. The L-MMSE combining scheme is applied. ``Tailored sparse weighting" uses the antenna/array activation patterns inferred from the optimized sparse weighting vectors to prune the initial OBE vectors and ``original sparse weighting" directly employs the optimized sparse weighting vectors for performance evaluation.
\label{Fig_Tailor_DiffN}}
\vspace{-0.1cm}
\end{figure}

\begin{figure}[t]
\centering
\includegraphics[width=0.8\columnwidth]{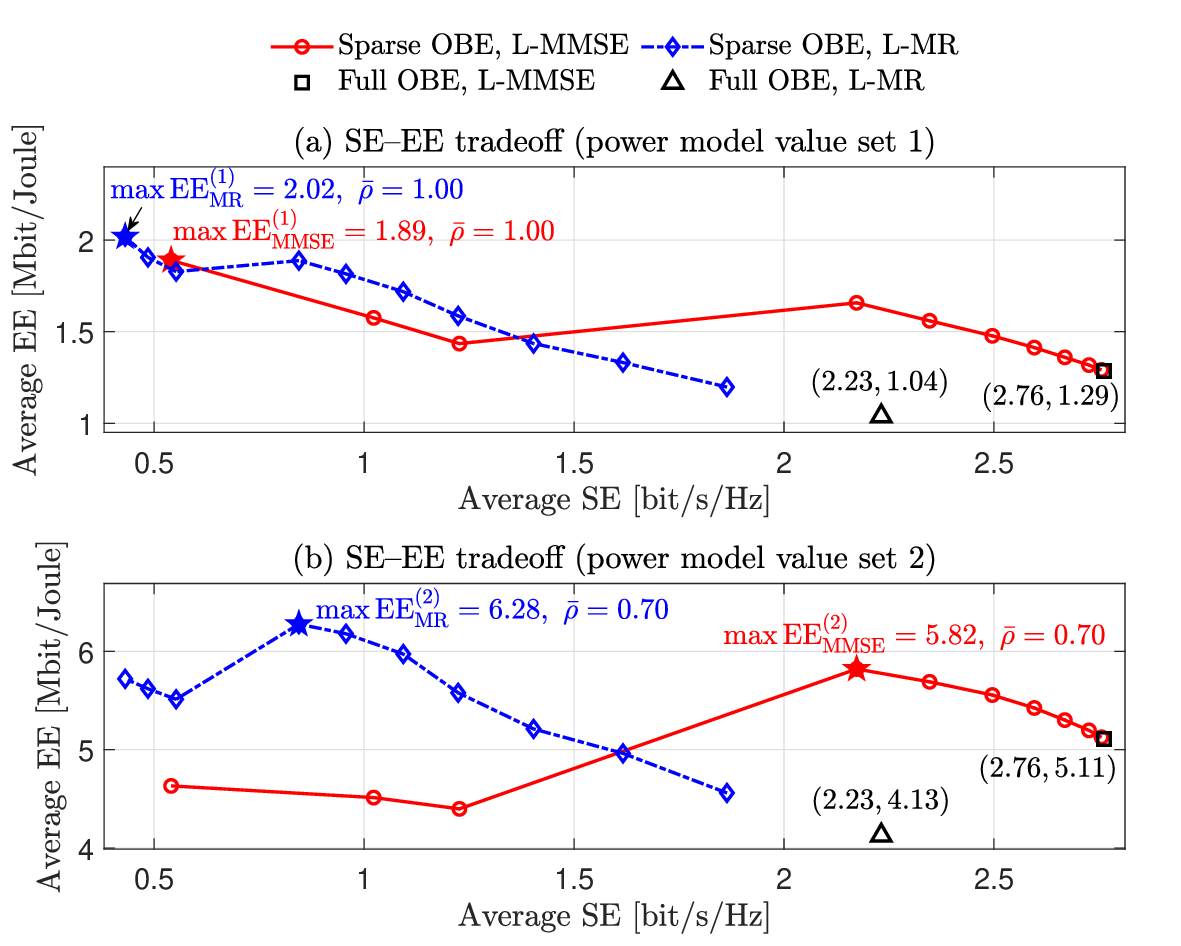}\vspace{-0.3cm}
\caption{Average SE and average system EE for the array/UE activation scheme over two different power consumption model value sets with $M=10$, $K=10$, and $N=9$. The value set $1$ denotes the set introduced in the beginning of this section and the value set $2$ denotes the set with $P_{k}^{\mathrm{ue},\mathrm{c}}=0.05\, \mathrm{W}$, $p_{k,p}=p_k=0.2 \,\mathrm{W}$, $\eta _{\mathrm{ue}}=0.6$,  $P_{m}^{\mathrm{array}}=2 \, \mathrm{W}$, $P_{mn}^{\mathrm{ant},\mathrm{c}}=0.1\, \mathrm{W}$, $P_{mnk}^{\mathrm{pro}}=0.08 \,\mathrm{W}$, $P_{m}^{\mathrm{fh},\mathrm{fix}}=0.2 \, \mathrm{W}$, $P_{mk}^{\mathrm{sig}}=2\, \mathrm{mW}$, $P^{\mathrm{cpu},\mathrm{fix}}=1 \,\mathrm{W}$, and $P^{\mathrm{cpu},\mathrm{dec}}=0.16 \, \mathrm{W}/\left( \mathrm{Gbit}/\mathrm{s} \right) $. \label{Fig_Power_Model}}
\vspace{-0.3cm}
\end{figure}

\begin{figure}[t]
\centering
\includegraphics[width=0.8\columnwidth]{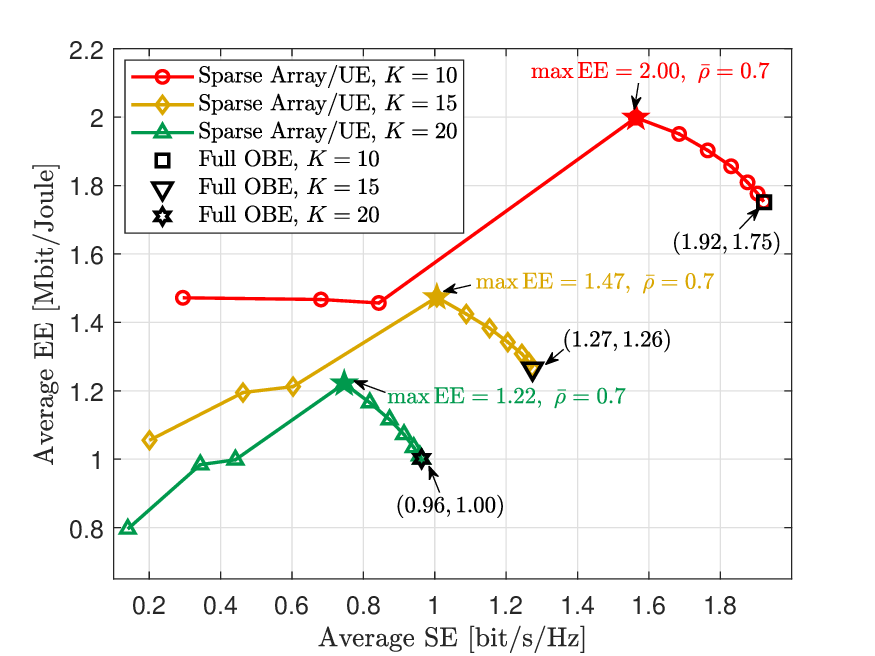}\vspace{-0.3cm}
\caption{Average SE and EE for the array/UE activation scheme with the L-MMSE combining applied over different $K$ with $M=10$ and $N=4$. \label{Fig_Diff_K}}
\vspace{-0.2cm}
\end{figure}

Fig.~\ref{Fig_array_all} illustrates the average SE per UE and average system EE for the sparse array-level activation scheme for all UEs in Sec.~\ref{array_all_UE}. As observed, the average EE reaches its maximum in the high-sparsity region, indicating that the power saving induced by array-level inactivation for this scheme is sufficiently significant to dominate the SE loss under moderate-to-strong sparsity. Meanwhile, the maximized EE values for both combining schemes are the highest among the maximum values in all considered activation schemes, which demonstrates that this array-level sparsification provides the most effective power consumption reduction and EE enhancement. Similar to the observation to Fig.~\ref{Fig_antenna_all}, we observe that the SE and EE variations are considerably more sensitive to the leaf sparsity ratio $ \bar{\rho} $, because leaf sparse factors $ \bar{\rho} $ or $ \bar{\lambda} $ controls the UE-specific antenna/array associations at a much finer granularity than the root sparse factor $ \bar{\omega} $, and thus reacts more strongly to changes in the sparsity level.

In Fig.~\ref{Fig_compare_all}, we compare all sparse activation schemes under a unified EE-SE tradeoff (EST) metric. Since the sparse factors that maximize SE and EE differ markedly across schemes, for comparison fairness and SE-EE balance, we adopt the EST measure from \cite[Eq. (10)]{6365872} with SE preference weight $\alpha = 0.3$, corresponding to an EE-oriented operating point. For each scheme, SE and EE are first normalized by their own maxima over all sparse factors, and the sparsity level that maximizes EST is then selected as the operating point used for comparison. The results show that, except for the conservative antenna/UE activation with the L-MMSE, all sparse schemes yield a substantial EE improvement over the full-activation benchmark. The antenna/UE scheme with the L-MMSE attains the highest SE among all sparse schemes and remains very close to the full-activation SE, while only slightly reducing the total power consumption. By contrast, the array/all UEs scheme achieves the largest EE and the minimum total power at the expense of a substantial SE loss. Moreover, the array/UE and antenna/all UEs schemes offer an intermediate and more balanced SE-EE tradeoff. As observed from Fig.~\ref{Fig_compare_all}(b), the antenna/UE scheme is the most conservative in terms of resource deactivation, retaining the largest number of active antennas/arrays, whereas the antenna/all UEs and array/all UEs schemes shut down considerably more resources than the UE-specific schemes. Moreover, for any given activation scheme, the L-MR combining keeps fewer active resources at its EST-optimal point than the L-MMSE, indicating that the L-MR is more sensitive to the introduction of sparsity.

Fig.~\ref{Fig_power_save} illustrates the SE degradation versus the power saving for the proposed sparse activation schemes. As observed, all schemes exhibit a clear SE--power-saving tradeoff, that is, a larger power saving is generally achieved at the expense of a larger SE loss. Among the four schemes, the array/all-UEs scheme attains the most favorable tradeoff boundary, since it consistently yields the smallest SE degradation for a given power saving, or equivalently achieves a larger power saving for the same SE loss. By contrast, the antenna/UE scheme is overall the least efficient, indicating that overly fine-grained activation leads to a less favorable SE--power-saving balance. Meanwhile, the array/UE scheme and the antenna/all UEs scheme exhibit rather comparable performance over the considered range, suggesting that both the activation granularity and the coordination scope play important roles in shaping the final tradeoff. Overall, these results show that coordinating the sparse activation across all UEs is generally more effective than designing it separately for each UE, and that array-level activation provides a more favorable SE--power-saving tradeoff than antenna-level activation.

\vspace{-0.3cm}
\subsection{Impact of System and Modeling Settings}

Fig.~\ref{Fig_Tailor_DiffN} investigates the impact of $N$ and the choice of sparse weighting schemes on the achievable performance for the sparse activation, where 
the array/UE scheme is taken as an example. We observe that the tailored sparse weighting consistently outperforms the original sparse weighting in both the SE and EE performance for all considered values of $N$ and the performance gap between the original one and the tailored one becomes large as $N$ increases, demonstrating the necessity of utilizing this tailored sparse weighting strategy when evaluating the performance. Since the two sparse schemes share almost identical power consumption except for the throughput–dependent CPU power, the EE degradation of the original sparse weighting mainly originates from its poor SE performance. These observations are because, in the original scheme, all entries of the optimized sparse weighting vectors are disturbed, including those that remain nonzero, whereas the tailored scheme only uses the extracted antenna/array activation modes to prune the initial OBE weights and preserves the remaining coefficients. Moreover, the SE improvement from more antennas is accompanied by the gradual decrease in EE due to the significantly growing power consumption, and the tailored sparse activation involves robustness as $N$ varies.

Fig.~\ref{Fig_Power_Model} investigates how different power consumption value sets affect the sparse activation scheme, again using the array/UE scheme as an example. The first set corresponds to the baseline values used throughout this paper, whereas the second set is obtained by reducing the hardware-related power parameters by a factor of two and the computation-related parameters by a factor of five, following the guideline in \cite{8187178}, so that value set~2 represents a more lightweight architecture. The results show that the proposed sparse activation scheme can substantially improve EE under both value sets. However, with value set~2 in Fig.~\ref{Fig_Power_Model}.(b), the SE-EE curves become clearly unimodal and the EE optimum shifts to a moderate sparsity level around $\bar{\rho}=0.7$, which contrasts with the behavior under value set~1, where the EE keeps increasing toward the most aggressive sparsity. These observations indicate that both the achievable EE and the optimal sparsity level of the sparse activation scheme are highly sensitive to the adopted power consumption model, and the value set should therefore be flexibly chosen according to realistic hardware parameters or specific design requirements, such as hardware power-dominant or processing power-dominant scenarios.

Fig.~\ref{Fig_Diff_K} investigates the impact of the user load on the sparse activation scheme. We take the array-level activation scheme for the particular UE as an example. As $K$ increases, both the full-activation benchmark and the sparse trade-off curves shift to the lower-left region, indicating that a heavier user load reduces both the average SE per UE and the system EE. Nevertheless, the proposed sparse array/UE activation remains effective over the whole tested load range. In particular, the EE-maximizing operating point consistently outperforms the full-activation OBE benchmark in terms of EE, e.g. $14.3\%$, $16.7\%$, and $22.0\%$ EE improvements for $K=10$, $15$, and $20$, respectively, at the expense of about $20\%$ SE reductions. These results show that sparse activation is not limited to very low user loads. Even at higher loads, deactivating less important links remains effective in reducing the power consumption and enhancing the EE while preserving a reasonable SE level.

\vspace{-0.15cm}

\section{Conclusions}
\vspace{-0.15cm}
The sparse antenna/array activation problems in uplink CF mMIMO networks were investigated. An antenna-level OBE weighting framework was developed, where each AP-UE pair was assigned a matrix-valued long-term weight to shape the contributions of individual antenna elements.
Based on this framework, the sparse activation problem was constructed as a structured sparsity-inducing MSE minimization problem, yielding four activation schemes at antenna-level and array-level, each with UE-specific and network-wide variants. The resulting convex formulations were solved using the proximal method with closed-form group-wise updates. Notably, the network-wide designs were characterized by hierarchical sparsity and handled via a tree-structured proximal operator.

Through this paper, we developed insights into the studied sparse activation schemes. \emph{First}, if the primary objective is to maximize EE or minimize the total power consumption even at the expense of SE loss, the array/all UEs scheme is the most suitable choice. \emph{Second}, the antenna/all UEs and array/UE schemes provide a more balanced SE-EE tradeoff, achieving considerable EE gains while maintaining a moderate SE degradation. 
\emph{Third}, the antenna/array activation schemes for all UEs are more effective in inactivating antenna/array resources than the UE-specific activation schemes, and the leaf sparse factors are more dominant in affecting performance than the root factors. \emph{Fourth}, for any given activation scheme, the L-MR combining is more sensitive to the introduction of sparsity than L-MMSE, allowing more antenna/array resources to be switched off while achieving a higher EE. 
\emph{Fifth}, the implementation of sparse activation schemes is highly dependent on the choice of value set in the power consumption model, which should be highlighted based on specific research requirements.

% \begin{figure}[t]
% \centering
% \includegraphics[width=0.82\columnwidth]{FIG_OBE_M.eps}
% \caption{xxxxx. \label{xxxx}}
% \end{figure}

\begin{appendices}
\vspace{-0.3cm}
\section{Proof of Proposition~\ref{prop_obe_weighting}}\label{app_obe_weighting}
We begin from the numerator of \eqref{SINR_k}. We can formulate $\sum_{m=1}^M{\mathbb{E} \{\mathbf{v}_{mk}^{H}\mathbf{h}_{mk}\}}\overset{\left( a \right)}{=}\sum_{m=1}^M{\mathbf{w}_{mk}^{H}\mathrm{vec}\left( \mathbb{E} \{\mathbf{h}_{mk}\overline{\mathbf{v}}_{mk}^{H}\} \right)}=\mathbf{w}_{k}^{H}\mathbb{E} \{\mathbf{g}_k\}$, where step (a) follows from the standard results as \cite[Eq. (32)]{OBETrans} and \cite[Eq. (34)]{OBETrans}.
% Meanwhile, we have $\mathbf{w}_k=\left[ \mathbf{w}_{1k}^{T},\mathbf{w}_{2k}^{T},\dots ,\mathbf{w}_{Mk}^{T} \right] ^T\in \mathbb{C} ^{MN^2}$ and $\mathbf{g}_k=[\mathrm{vec(}\mathbf{h}_{1k}\overline{\mathbf{v}}_{1k}^{H})^T,\dots ,\mathrm{vec(}\mathbf{h}_{Mk}\overline{\mathbf{v}}_{Mk}^{H})^T]^T\in \mathbb{C} ^{MN^2}$. 
As for $\mathbb{E} \{|\sum_{m=1}^M{\mathbf{v}_{mk}^{H}\mathbf{h}_{ml}|^2\}}$, we have $\mathbb{E} \{|\sum\nolimits_{m=1}^M{\mathbf{v}_{mk}^{H}\mathbf{h}_{ml}|^2\}}\overset{\left( a \right)}{=}\mathbb{E} \{|\sum\nolimits_{m=1}^M{\mathrm{tr}(\mathbf{W}_{mk}^{H}\mathbf{h}_{ml}\overline{\mathbf{v}}_{mk}^{H})|^2\}}\overset{\left( b \right)}{=}\mathbb{E} \{|\sum\nolimits_{m=1}^M{\mathbf{w}_{mk}^{H}\mathrm{vec(}\mathbf{h}_{ml}\overline{\mathbf{v}}_{mk}^{H})|^2\}}=\mathbb{E} \{|\mathbf{w}_{k}^{H}\mathbf{z}_{kl}|^2\}=\mathbf{w}_{k}^{H}\mathbb{E} \{\mathbf{z}_{kl}\mathbf{z}_{kl}^{H}\}\mathbf{w}_k$,
where step (a) and step (b) follow from \cite[Eq. (32)]{OBETrans} and \cite[Eq. (34)]{OBETrans}, respectively. Finally, we can construct $\sum_{m=1}^M{\mathbb{E} \{\parallel \mathbf{v}_{mk}\parallel ^2\}=\sum_{m=1}^M{\mathrm{tr(}\mathbf{W}_{mk}^{H}\mathbf{W}_{mk}\mathbb{E} \{\overline{\mathbf{v}}_{mk}\overline{\mathbf{v}}_{mk}^{H}\})}}\overset{\left( a \right)}{=}\sum\nolimits_{m=1}^M{\mathbf{w}_{mk}^{H}(\mathbb{E} \{\overline{\mathbf{v}}_{mk}\overline{\mathbf{v}}_{mk}^{H}\}^T\otimes \mathbf{I}_N)\mathbf{w}_{mk}}=\mathbf{w}_{k}^{H}\mathbf{\Theta }_k\mathbf{w}_k$, where step (a) follows from \cite[Eq. (36)]{OBETrans}. Based on the above results, we can construct $\mathrm{SINR}_k$ in \eqref{SINR_k} as
\vspace{-0.55em}\begin{equation}\label{SINR_Construct}
\begin{aligned}
&\mathrm{SINR}_k=\\
&\frac{p_k|\mathbf{w}_{k}^{H}\mathbb{E} \{\mathbf{g}_k\}|^2}{\mathbf{w}_{k}^{H}(\sum_{l=1}^K{p_l\mathbb{E} \{\mathbf{z}_{kl}\mathbf{z}_{kl}^{H}\}-p_k\mathbb{E} \{\mathbf{g}_k\}\mathbb{E} \{\mathbf{g}_k\}^H+\sigma ^2\mathbf{\Theta }_k)\mathbf{w}_k}}.
\end{aligned}
\vspace{-0.3em}\end{equation}
We can observe that \eqref{SINR_Construct} is a Rayleigh quotient with respect to $\mathbf{w}_k$, and thus, following \cite[Lemma B.10]{8187178} and steps as \cite[C.3.2]{8187178}, we derive the optimal $\mathbf{w}_k$ as \eqref{opt_weighting}.

\vspace{-0.3cm}
\section{Proof of MSE Minimization by OBE Weighting}\label{MSE_Minimization}
When applying the weighted local decoding scheme defined in \eqref{local_decoding_scheme}, we can easily construct the final decoding data in \eqref{decoding_signal} as $\check{x}_k=\mathbf{w}_{k}^{H}\mathbf{g}_kx_k+\sum_{l\ne k}^K{\mathbf{w}_{k}^{H}\mathbf{z}_{kl}x_l}+\sum_{m=1}^M{\overline{\mathbf{v}}_{mk}^{H}\mathbf{W}_{mk}^{H}\mathbf{n}_m}$, where the mentioned variables are defined in Proposition~\ref{prop_obe_weighting}. Then, we can easily derive
$\mathrm{MSE}_k=\mathbb{E} \{ | x_k-\check{x}_k |^2\} =p_k-p_k\mathbb{E} \{ \mathbf{g}_{k}^{H} \} \mathbf{w}_k-p_k\mathbf{w}_{k}^{H}\mathbb{E} \{ \mathbf{g}_k \} +
\mathbf{w}_{k}^{H}( \sum_{l=1}^K{p_l\mathbb{E} \{ \mathbf{z}_{kl}\mathbf{z}_{kl}^{H} \}} ) \mathbf{w}_k+\sigma ^2\mathbf{w}_{k}^{H}\mathbf{\Theta }_k\mathbf{w}_k$ by utilizing the random characteristics of $\{x_k:\forall k\}$. Then, by computing the partial derivation of $\mathrm{MSE}_k$ with respect to $\mathbf{w}_{k}^{*}$, we derive the optimal scheme, minimizing $\mathrm{MSE}_k$, as $\mathbf{w}_k=p_k( \sum_{l=1}^K{p_l\mathbb{E} \{ \mathbf{z}_{kl}\mathbf{z}_{kl}^{H} \}}+\sigma ^2\mathbf{\Theta }_k) ^{-1}\mathbb{E} \{ \mathbf{g}_k \} $, which is same to the OBE weighting scheme as in \eqref{opt_weighting}.

\end{appendices}

\vspace{-0.3cm}
\bibliographystyle{IEEEtran}
\bibliography{IEEEabrv,Ref}

\end{document}